\documentclass{article}

    \PassOptionsToPackage{numbers, compress}{natbib}

\usepackage[preprint]{neurips_2026}

\usepackage[utf8]{inputenc} %
\usepackage[T1]{fontenc}    %
\usepackage{url}            %
\usepackage{booktabs}       %
\usepackage{amsfonts}       %
\usepackage{nicefrac}       %
\usepackage{microtype}      %
\usepackage{xcolor}         %

\usepackage{graphicx}
\usepackage{subcaption}
\usepackage{amsmath}
\usepackage{amssymb}
\usepackage{mathtools}
\usepackage{amsthm}
\usepackage{hyperref}       %
\usepackage{booktabs}
\usepackage{colortbl}
\usepackage{xcolor}
\usepackage{tabularx}  
\usepackage{etoolbox}
\usepackage{multirow}
\usepackage[capitalize,noabbrev]{cleveref}
\usepackage[textsize=tiny]{todonotes}
\usepackage{wrapfig}
\usepackage{float}
\usepackage{algorithm}
\usepackage{algorithmic}
\usepackage{aliascnt}

\theoremstyle{plain}

\newaliascnt{proposition}{theorem}

\aliascntresetthe{proposition}
\crefname{proposition}{Proposition}{Propositions}

\newaliascnt{lemma}{theorem}
\newtheorem{lemma}[lemma]{Lemma}
\aliascntresetthe{lemma}
\crefname{lemma}{Lemma}{Lemmas}

\newaliascnt{corollary}{theorem}

\aliascntresetthe{corollary}
\crefname{corollary}{Corollary}{Corollaries}

\theoremstyle{definition}
\newaliascnt{definition}{theorem}
\newtheorem{definition}[definition]{Definition}
\aliascntresetthe{definition}
\crefname{definition}{Definition}{Definitions}
\newaliascnt{assumption}{theorem}
\newtheorem{assumption}[assumption]{Assumption}
\aliascntresetthe{assumption}
\crefname{assumption}{Assumption}{Assumptions}

\theoremstyle{remark}

\usepackage{thmtools}
\usepackage{thm-restate}

\title{MUTE: Return-Preserving Communication Unlearning for Efficient Multi-Agent Coordination}

\author{%
  Rui Zuo \\
  Syracuse University\\
  \texttt{rzuo02@syr.edu} \\
  \And
  Qinwei Huang \\
  Syracuse University\\
  \texttt{qhuang18@syr.edu} \\
  \And
  Mingyang Li \\
  Syracuse University\\
  \texttt{mli170@syr.edu} \\
  \And
  Zhenhang Zhang \\
  Syracuse University\\
  \texttt{zzhan281@syr.edu} \\
  \And
  Simon Khan \\
  Air Force Research Laboratory\\
  \texttt{simon.khan@us.af.mil} \\
  \And
  Qinru Qiu \\
  Syracuse University\\
  \texttt{qiqiu@syr.com} \\
}

\newlength\savewidth\newcommand\shline{\noalign{\global\savewidth\arrayrulewidth
  \global\arrayrulewidth 1pt}\hline\noalign{\global\arrayrulewidth\savewidth}}

\newcommand{\tablestyle}[2]{\setlength{\tabcolsep}{#1}\renewcommand{\arraystretch}{#2}\centering\footnotesize}
\definecolor{baselinecolor}{gray}{.9}

\newcolumntype{C}[1]{>{\centering\arraybackslash}p{#1pt}}
\newcolumntype{L}[1]{>{\raggedright\arraybackslash}p{#1pt}}
\newcolumntype{R}[1]{>{\raggedleft\arraybackslash}p{#1pt}}

\usepackage{xspace}
\makeatletter
\DeclareRobustCommand\onedot{\futurelet\@let@token\@onedot}
\def\@onedot{\ifx\@let@token.\else.\null\fi\xspace}

\makeatother

\colorlet{darkgreen}{green!65!black}
\colorlet{darkblue}{blue!75!black}
\colorlet{darkred}{red!80!black}
\definecolor{lightblue}{HTML}{0071bc}
\definecolor{lightgreen}{HTML}{39b54a}

\definecolor{bad}{RGB}{255,235,235}    %
\definecolor{good}{RGB}{235,255,235}    %

\definecolor{lightgrey}{gray}{0.8}

\newcommand{\gline}{\arrayrulecolor{lightgrey}\hline\arrayrulecolor{black}}

\newcommand{\btau}{\boldsymbol{\tau}}

\begin{document}

\maketitle
\addtocontents{toc}{\protect\setcounter{tocdepth}{0}}

\begin{abstract}
Inter-agent communication is critical for coordinating Multi-Agent Reinforcement Learning (MARL) agents under partial observability to perform effectively in cooperative games; however, real-world bandwidth constraints demand sparse interactions.
Prior approaches primarily address this trade-off by optimizing information-theoretic surrogates.
We argue that these statistical proxies are fundamentally misaligned with the true objective: a message can be highly informative yet irrelevant to the joint return of the task.
In this work, we propose \textbf{M}essage \textbf{U}nlearning for \textbf{T}argeted \textbf{E}fficiency (\textbf{MUTE}), a framework that views communication reduction as a value-guided machine unlearning problem.
MUTE rigorously quantifies the \textit{Counterfactual Message Value} using an attention-based estimator, and systematically unlearns the transmission of low-value messages from a policy trained without any communication constraints. This is achieved through a dual-objective mechanism that enforces communication sparsity while preserving the return of the original joint policy.
We derive a theoretical upper bound on the performance gap induced by this sparsification, guaranteeing controlled return degradation.
We also empirically evaluate MUTE on various complex multi-agent environments, achieving 80\% to 90\% bandwidth reduction while maintaining performance comparable to state-of-the-art baselines.
\end{abstract}

\section{Introduction}
\label{sec:intro}

Multi-Agent Reinforcement Learning (MARL) has emerged as a powerful paradigm for solving complex coordination problems, ranging from autonomous vehicles~\cite{marl_autonomous_vehicles} to micromanagement in StarCraft II \cite{smac}. To address the challenges of non-stationarity and scalability, the Centralized Training with Decentralized Execution (CTDE)~\cite{ctde} framework has become the standard, exemplified by value-based methods~\cite{qmix,maic} and actor-critic approaches~\cite{mappo,maddpg,dial}.
However, in environments characterized by partial observability and policy-induced non-stationarity, decentralized execution often falters. Communication emerges as a pivotal solution to this challenge \cite{dial,commnet}, empowering agents to transcend the limitations of local perception by exchanging diverse signals—ranging from raw observations~\cite{ic3net, masia} and latent intentions~\cite{maic} to encoded experiences~\cite{t2mac}. This continuous information exchange serves to stabilize the learning process against environmental non-stationarity, granting agents a clearer understanding of the global state~\cite{maic,ndq}.

While unrestricted communication can mitigate partial observability, bandwidth remains a scarce resource in real-world deployments~\cite{imac, marl_comm_survey}. 
Prior works have primarily addressed this constraint by optimizing information-theoretic~\cite{imac, ndq, maic} auxiliary objectives to reduce communication.
For instance, IMAC~\cite{imac} enforces communication sparsity via the Information Bottleneck principle, explicitly \textit{minimizing} the mutual information between observations and messages, which implicitly assumes that high-entropy messages are redundant.
NDQ~\cite{ndq} and MAIC~\cite{maic} prioritize messages by \textit{maximizing} the mutual information between the message and the recipient's action selection.
This promotes messages that alter behavior, regardless of whether that alteration leads to higher returns.
By optimizing for statistical significance without directly targeting the joint return, these methods risk prioritizing messages that are merely informative over those that are truly return-critical.

In this work, we reframe communication reduction as a return-optimization problem, directly learning which messages are critical for maximizing team performance.
Although SMS~\cite{sms} shares this objective, it utilizes Shapley values~\cite{shapley} to estimate each message's marginal contribution  to the expected return. However, exact Shapley value computation requires evaluating all possible message subsets, resulting in exponential computational complexity. Consequently, SMS is forced to rely on sampling approximations that exhibit high-variance, and restrict the value function to a linear decomposition~\cite{dop}, which severely limits its applicability to complex tasks requiring non-monotonic coordination.
In contrast, we propose a more tractable approach to estimate the marginal contribution of a message via \emph{Counterfactual Message Value} (CMV), which captures the difference in the joint action-value function between two counterfactual scenarios: when the message is transmitted versus when it is masked. 
We train a Message-Value Estimator (MVE) to approximate the CMV, and apply the MVE to identify the subset of redundant messages whose removal does not degrade the expected joint return relative to a policy with unrestricted communication.

To efficiently remove the identified redundancies without destabilizing learned coordination, we propose \textit{Message Unlearning for Targeted Efficiency} (MUTE).
Inspired by machine unlearning~\cite{unlearn_original, unlearn1, unlearn_survey}, which seeks to erase the influence of specific training data from trained models, we view communication reduction as a targeted unlearning task that progressively \textit{unlearns} the transmission of low-value messages identified by our estimator.
To ensure reliability, we derive a theoretical upper bound on the performance gap induced by unlearning, guaranteeing that any return degradation due to discarding messages is rigorously bounded within a controllable margin.
We further design a dual-objective loss consisting of (i) a sparsity term that penalizes low-value messages and (ii) an anchoring term that prevents policy collapse by regularizing the agents' behavior toward the original policy.

We evaluate MUTE on diverse benchmarks, including Hallway~\cite{ndq}, StarCraft Multi-Agent
Challenge (SMAC)~\cite{smac}, Traffic Junction (TJ)~\cite{ic3net}, SMAC-Communication~\cite{ndq} and SMACv2~\cite{smacv2}.
Empirical results demonstrate that MUTE achieves performance comparable to state-of-the-art approaches with up to 99\% communication reduction.

\section{Related Work}

\begin{wraptable}{r}{0.35\textwidth}
\centering
\tablestyle{4pt}{1.2}
\begin{tabular}{L{70} | L{50}}
Category & Method \\
\shline
\multirow{2}{*}{\textbf{No Communication}} & QMIX~\cite{qmix}\\
& MAPPO~\cite{mappo} \\
\gline
\textbf{Unrestricted} & MASIA~\cite{masia} \\
\gline
\multirow{2}{*}{\textbf{Mutual Information}} & NDQ~\cite{ndq} \\
& IMAC~\cite{imac} \\
\gline
\textbf{Gating} & IC3Net~\cite{ic3net} \\
\gline
\multirow{2}{*}{\textbf{Attention}} & TarMAC~\cite{tarmac} \\
& MAIC~\cite{maic} \\
\gline
\textbf{Shapley Value} & SMS~\cite{sms} \\
\gline
\rowcolor{good} \textbf{Unlearning (ours)} & \textbf{MUTE} \\
\end{tabular}
\caption{\textbf{Comparison of Multi-agent Communication baselines.}}
\label{tab:compare}
\vspace{-2em}
\end{wraptable}

\paragraph{Multi-agent Communication}
In the paradigm of CTDE, agents learn to coordinate without explicit communication.
Value-based methods like QMIX~\cite{qmix} have shown success by factoring the joint value function, while MAPPO~\cite{mappo} has demonstrated strong performance using a centralized critic approach.
However, because these methods rely solely on implicit coordination derived from training history, they often struggle under partial observability, where agents lack the local information necessary to infer the global state.

To overcome this informational bottleneck, explicit communication protocols have been introduced.
Frameworks such as TarMAC~\cite{tarmac} and MAIC~\cite{maic} utilize soft attention mechanisms to weigh incoming messages, allowing agents to selectively focus on relevant information from their teammates.
Building on this, T2MAC~\cite{t2mac} further optimizes efficiency by learning to explicitly target specific recipients rather than broadcasting to all.
MASIA~\cite{masia} extends these ideas by integrating unrestricted communication into a value-based framework.
Although these methods significantly improve coordination, they often ignore bandwidth constraints, leading to redundant information processing and high communication overhead.

Several works have been proposed to mitigate these costs.
IC3Net~\cite{ic3net} employs a gating mechanism that allows agents to learn when to block communication based on local context.
Information-theoretic approaches like NDQ~\cite{ndq} and IMAC~\cite{imac} leverage the Information Bottleneck principle to compress messages by minimizing mutual information while preserving value prediction.
Similarly, SMS~\cite{sms} adopts a game-theoretic perspective, estimating the Shapley value of messages to prune those with low marginal contribution.
Distinct from these methods, which often struggle to balance sparsity and performance during training, our MUTE framework introduces a novel \textit{unlearning} paradigm: we initialize with a policy without any communication constraints and progressively prune redundant channels via behavioral anchoring, ensuring efficiency without sacrificing cooperative returns.

\section{Problem Formulation}

We consider a fully cooperative multi-agent communication task modeled as a Decentralized Partially Observable Markov Decision Process (Dec-POMDP), defined by the tuple $\mathcal{G} = \langle \mathcal{N}, \mathcal{S}, \mathcal{A}, P, r, \Omega, O, \gamma, \mathcal{M} \rangle$.
$\mathcal{N} = \{1, \dots, N\}$ denotes the set of $N$ agents.
$\mathcal{S}$ represents the global state of the environment.
$\mathcal{A}$, $\Omega$, and $\mathcal{M}$ represent the action, observation, and message spaces, respectively, while $O$ is the observation function. We use $\tau_i^t$ to denote the local action-observation history of agent $i$ up to time $t$, $\tau_i^t = (o_i^0, a_i^0, \dots, a_i^{t-1}, o_i^t)$.

We adopt a broadcast communication protocol.
At each timestep $t$, each agent $i$ generates a message $m_i^t$ using a learnable message generator $f_{\psi}$, conditioned on its local action-observation history: $m_i^t = f_{\psi}(\tau_i^t)$.
Under the broadcast setting, this message is transmitted identically to all other agents.
Consequently, agent $i$ receives a set of messages from its teammates, denoted as $\mathbf{m}_{-i}^t = \{m_j^t\}_{j \in \mathcal{N} \setminus \{i\}}$.
After the communication phase, each agent selects an action $a_i \in \mathcal{A}$.
The agent's policy $\pi(a_i | \tau_i^t, \mathbf{m}_{-i}^t)$ is conditioned on its local action-observation history $\tau_i^t$ and the set of received messages.

The joint action $\mathbf{a} = \langle a_1, \dots, a_n \rangle$ triggers a transition to the next state $s' \sim P(\cdot | s, \mathbf{a})$ and yields a global reward $r(s, \mathbf{a})$.
The objective is to learn a joint policy $\boldsymbol{\pi}$ that maximizes the expected discounted return $J(\boldsymbol{\pi}) = \mathbb{E} \left[\sum_{t=0}^{\infty} \gamma^t r_t \right] = \mathbb{E}_{\boldsymbol{\tau}, \mathbf{a}} [Q_{\text{tot}}(\boldsymbol{\tau}, \mathbf{a}; \theta)]$.
In value decomposition methods~\cite{qmix}, the joint action-value function $Q_{\text{tot}}$, parametrized by $\boldsymbol{\theta}$, is composed of individual utility functions $Q_i(\tau_i, a_i; \theta)$, which determine the decentralized policies.

We adopt the CTDE paradigm.
In the centralized RL training phase, the network is trained end-to-end by minimizing the standard Temporal Difference (TD) error:
\begin{equation}
    \mathcal{L}(\theta) = \mathbb{E}_{\mathcal{D}} \left[ \left( y^{\text{tot}} - Q_{\text{tot}}(\boldsymbol{\tau}, \mathbf{a}; \theta) \right)^2 \right]
    \label{eq:td_loss}
\end{equation}
where $y^{\text{tot}} = r + \gamma \max_{\mathbf{a}} Q_{\text{tot}}(\boldsymbol{\tau}', \mathbf{a}; \theta^-)$ is the target value, computed using the target network parameters $\theta^-$, and $\tau'$ represents the action-observation history at state $s'$.
We denote the converged joint action-value function as $Q^*_{\text{tot}}$ and the converged individual utility functions as $Q_i^*$.
For simplicity, we omit the time index $t$ hereafter.

\section{Method}

\subsection{Revisiting Communication Reduction: A Return-Preserving Perspective}

As summarized in \cref{tab:compare}, existing approaches to communication reduction largely rely on penalizing message generation during policy learning.
Whether through Gating~\cite{ic3net}, Attention~\cite{tarmac}, or Mutual Information regularization~\cite{ndq, imac}, these methods fundamentally treat communication as a cost to be balanced against environmental returns.
This formulation forces agents to simultaneously learn \textit{how} to cooperate (maximize return) and \textit{when} to remain silent (minimize communication cost), often leading to unstable convergence or suboptimal policies where agents sacrifice critical coordination to minimize penalties~\cite{karten2023towards, wei2025rescom}.

To empirically verify whether simultaneously optimizing task performance and minimizing communication cost compromises coordination, we conducted a preliminary experiment on the SMAC-Communication (\texttt{1o\_2r\_vs\_4r})~\cite{ndq} scenario, as shown in \Cref{fig:intro_exp}, where agents were trained with varying communication rates for 1M steps. The results reveal that the agents' performance degrades significantly as the communication rate decreases.
This empirical evidence supports our hypothesis that indiscriminate communication reduction compromises coordination.
\textit{This suggests that a better approach may be to first train the policy under unrestricted communication bandwidth, and then unlearn unnecessary communication to reach the bandwidth requirements.}

\begin{wrapfigure}{r}{0.45\textwidth}
\centering
\includegraphics[width=\linewidth]{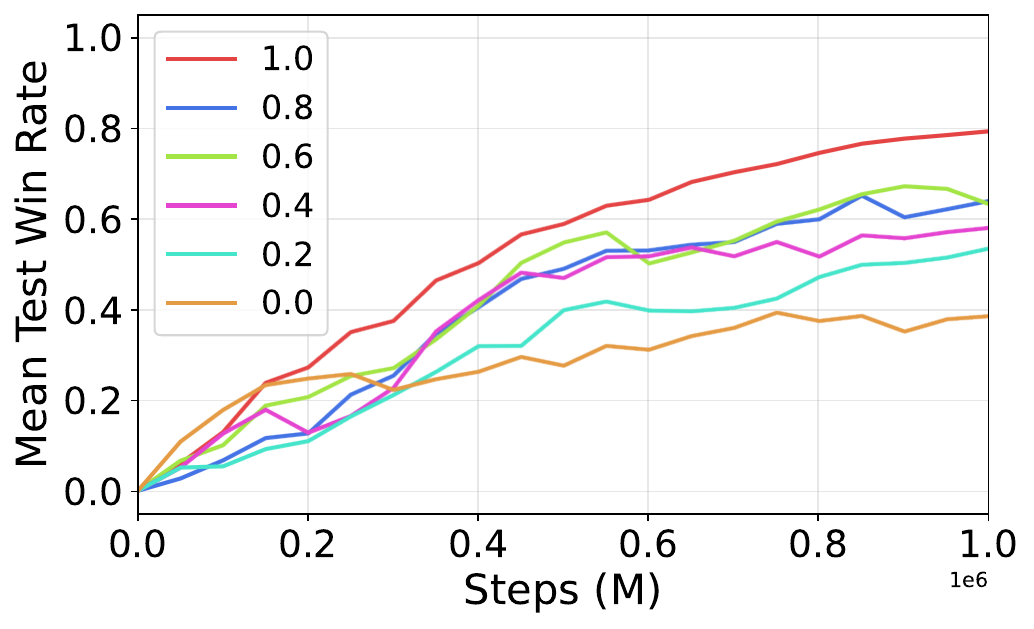}
\caption{Impact of communication reduction on performance during MARL training. The performance degradation highlights the compromise in coordination when communication is restricted indiscriminately.}
\label{fig:intro_exp}
\end{wrapfigure}

Moreover, as discussed in~\cref{sec:intro}, by optimizing information-theoretic objectives without directly targeting the joint return, these methods~\cite{ndq,maic,imac} risk prioritizing messages that are merely informative over those that are truly return-critical.
For example, consider agent $j$ whose local history $\tau_j$ alone is sufficient to select the optimal action $a_j^*$.
Suppose sender $i$ transmits a misleading message $m_{\text{bad}}$.
Influenced by this message, agent $j$ instead executes a suboptimal action $a_{\text{bad}}$, yielding a low reward $r_{\min}$.
In this case, the message strongly drives the action, producing high mutual information $\big(I(A_j; M_{i}) \gg 0\big)$ even as the return degrades.
To address the aforementioned challenges, we formulate communication reduction as a \textit{redundancy elimination} problem.
We posit that a fully communicative expert policy naturally contains a subset of return-critical messages sufficient to maintain (near-)optimal performance, mixed with redundant messages.
Therefore, instead of balancing competing objectives from scratch, our approach aims to \textit{identify and discard} these redundant messages while preserving the original returns.

Formally, our goal is to learn a sparse communication policy $\boldsymbol{\pi}$ that minimizes total communication overhead while rigorously bounding return gap:
\begin{equation}
\label{eq:constrained_opt}
\min_{\boldsymbol{\pi}} \mathbb{E} \left[ \sum_{t=0}^{T} \sum_{i=1}^{n} C(m_i^t) \right] \quad \text{s.t.} \quad J(\boldsymbol{\pi}^*) - J(\boldsymbol{\pi}) \le \epsilon
\end{equation}
where $C(m_i^t) \ge 0$ represents the message cost, $\boldsymbol{\pi}^{*}$ is the policy derived from the converged critic $Q_{\text{tot}}^*$ (i.e., greedy w.r.t $Q_{\text{tot}}^*$), serving as the performance upper bound, and $J(\boldsymbol{\pi})$ represents the expected discounted return under $\boldsymbol{\pi}$.

\subsection{Counterfactual Message Value Estimation}
\label{sec:cmv}

To solve the constrained optimization in \cref{eq:constrained_opt}, we must strictly distinguish between essential and redundant communication.
However, directly measuring the impact of a single message on the total return $J(\boldsymbol{\pi})$ is intractable due to the long time horizon.
In value-based MARL, the joint action-value function $Q_{\text{tot}}$ serves as a reliable approximation of the expected future return.
Leveraging this, we introduce \textit{Counterfactual Message Value} (CMV) to measure the value of message $m_{i}$ by comparing the global Q-value with and without that message.
This concept mirrors the counterfactual advantage function in COMA~\cite{coma}; however, instead of estimating the advantage of action, CMV measures the \textit{message's advantage}.

\begin{definition}[Counterfactual Message Value]
The counterfactual value of broadcast message $m_{i}$ is the difference in expected return caused by masking that specific message for all receivers.

\begin{equation}
    \Delta Q_{i}(\boldsymbol{\tau}, \mathbf{m}) = Q_{\text{tot}}^{*}(\boldsymbol{\tau}, \mathbf{a}) - Q_{\text{tot}}^{*}(\boldsymbol{\tau}, \tilde{\mathbf{a}})
    \label{eq:cmv}
\end{equation}
where $\mathbf{a} \sim \boldsymbol{\pi}^*(\cdot \mid \boldsymbol{\tau}, \mathbf{m})$ is the actual joint action sampled from the optimal policy,
and $\tilde{\mathbf{a}}$ is the counterfactual joint action where the actions of all teammates are resampled under the masked message:
\begin{equation}
    \tilde{\mathbf{a}} = (a_i, \tilde{\mathbf{a}}_{-i}) \quad \text{with} \quad \tilde{\mathbf{a}}_{-i} \sim \boldsymbol{\pi}^*_{-i}(\cdot \mid \boldsymbol{\tau}_{-i}, \mathbf{m} \setminus \{m_{i}\})
\end{equation}
Here, $a_i$ denotes the fixed action of the sender, and $\tilde{\mathbf{a}}_{-i}$ denotes the resampled actions of all other agents given that $m_i$ is masked.
For notational convenience, we denote the values of these induced joint actions as $Q_{\text{tot}}^*(\boldsymbol{\tau}, \mathbf{m}) \triangleq Q_{\text{tot}}^*(\boldsymbol{\tau}, \mathbf{a})$ and $Q_{\text{tot}}^*(\boldsymbol{\tau}, \mathbf{m}\setminus\{m_i\}) \triangleq Q_{\text{tot}}^*(\boldsymbol{\tau}, \tilde{\mathbf{a}})$.
Intuitively, $\Delta Q_{i}$ measures the return gap caused by silencing agent $i$.
If $\Delta Q_{i} \approx 0$, this specific signal is redundant given the current context and can be pruned without degrading performance.
Conversely, a large positive $\Delta Q_{i}$ indicates a return-critical signal that must be preserved.
\end{definition}

However, calculating the exact counterfactual baseline requires marginalizing over the receiver's policy and repeatedly querying the centralized critic $Q^*_{\text{tot}}$ for every communication channel for both masked and unmasked scenarios.
Therefore, estimating the CMV of all messages in a single communication round requires $O(N(N-1))$ forward passes.
This computational burden is prohibitively expensive, particularly for the downstream unlearning task, which demands efficient value estimation.

To address this inefficiency, we design the \textit{Message Value Estimator} (MVE) using an attention mechanism.
Crucially, this formulation allows us to predict the marginal values for all agents simultaneously in a single forward pass.
The choice of an attention architecture is motivated by two key properties.
First, the CMV must be \textit{permutation-invariant}, treating the incoming messages as an unordered set rather than a fixed sequence.
Second, the CMV is \textit{context-dependent}; for instance, if multiple agents broadcast identical information, the marginal value of any single message diminishes, as masking it does not result in information loss.
The self-attention mechanism naturally satisfies both requirements.

During the message value learning phase, the objective is:
\begin{equation}
    \mathcal{L}_{\text{MVE}}(\phi) = \frac{1}{N} \sum_{i=1}^{N} \left( V_\phi(\mathbf{m})_i - y_i \right)^2
    \label{eq:mve_loss}
\end{equation}
where $V_\phi(\mathbf{m})$ is the MVE parameterized by $\phi$, $V_{\phi}:\mathcal{M}\rightarrow \mathbb{R}^N$. %
The target $y_i$ represents the marginal contribution of agent $i$'s message, calculated using~\cref{eq:cmv}.
The MVE takes the joint set of all messages $\mathbf{m}$ as input. Because each message is generated directly from an agent's local observation history ($\mathbf{m} = f_\psi(\btau)$), $\mathbf{m}$ inherently encapsulates this vital state information. We verified this via an ablation study of a state-and-message variant, $V_\phi(s, \mathbf{m})$, which showed no performance gains and higher communication rates compared to our message-only model (full details are provided in the~\cref{app:mve_ablation}).

\subsection{Communication Unlearning}
To improve communication efficiency, we define the communication cost $C(\mathbf{m})$ in \cref{eq:constrained_opt} as the $L_1$ norm of low-value messages, yielding the sparsity loss:
\begin{equation}
    \mathcal{L}_{\text{sparse}}(\psi) = \sum_{m \in \mathcal{M}_{\text{red}}} \| m \|_1
    \label{eq:sparsity_loss}
\end{equation}
where $\mathcal{M}_{\text{red}} = \{ m_i \mid V_\phi(\mathbf{m})_i \le \lambda \mu_{\text{MVE}} \}$ denotes the set of messages whose estimated CMV falls below the adaptive sparsity threshold. Here, $\lambda$ is a scaling hyperparameter, and $\mu_{\text{MVE}}$ represents the exponential moving average of MVE predictions.
Rather than removing these messages outright, the $L_1$ penalty applies a continuous gradient pressure on their generators, gradually attenuating low-value channels while leaving the optimization free to revive any message whose CMV rises above $\lambda$ as the policy evolves.

Moreover, we approximate the return gap using \textit{Behavioral Anchoring}, which minimizes the divergence between the current policy $Q_{i}$ and the frozen baseline $Q_{i}^*$:
\begingroup
\small
\begin{equation}
    \mathcal{L}_{\text{anchor}}(\theta, \psi) = \sum_{i=1}^{N} \left\| Q_i^*(\tau_i, a_i, \mathbf{m}_{-i}) - Q_i(\tau_i, a_i, \mathbf{m}_{-i}; \theta) \right\|^2_2
    \label{eq:anchor_loss}
\end{equation}
\endgroup
By anchoring the individual value estimates rather than the global $Q_{\text{tot}}$, we ensure that each agent preserves its specific role and local policy stability, preventing the decentralized policies from collapsing as the communication channels are pruned.

Finally, we substitute the sparsity cost and the anchoring constraint into the Lagrangian objective.
This yields the final unlearning loss, optimized over both the policy parameters $\theta$ and the message generator parameters $\psi$:
\begin{equation}
    \mathcal{L}(\theta, \psi) = \mathcal{L}_{\text{sparse}}(\psi) + \beta \mathcal{L}_{\text{anchor}}(\theta, \psi)
    \label{eq:final_loss_combined}
\end{equation}
A subtle case arises when two messages carry the same useful information. In such a case, the CMV of each message may be near zero, since masking either one alone causes little or no return loss.
We show in \cref{app:duplicate_message} that the joint objective $\mathcal{L}$ optimizes toward keeping exactly one of the two messages, without jointly suppressing useful information.
The complete training procedure for MUTE is summarized in \cref{alg:mute_full} in~\cref{app:alg}.

To rigorously justify this intuition, we now provide a theoretical analysis quantifying the maximum performance degradation induced by our unlearning process.
\begin{restatable}[Performance Bound of Unlearning]{theorem}{ReturnGapThm}
\label{thm:return_gap}
Let $Q_{\text{tot}}^*$ and $Q_{\text{tot}}$ denote the joint value functions of the converged policy $\boldsymbol{\pi}^*$ and the unlearned policy $\boldsymbol{\pi}$, respectively.
Under~\cref{asm:smoothness}, and assuming MVE has a bounded error $| V_\phi(m) - \Delta Q(m) | \le \delta$, alongside an action-value approximation error bounded by $\|Q^*_{\text{tot}} - Q_{\text{tot}}\| \le \epsilon$, the performance gap is bounded by:
\begin{equation}
J(\boldsymbol{\pi}^*) - J(\boldsymbol{\pi}) \le \frac{1}{1-\gamma} \mathbb{E} \Bigg[ 2\epsilon + |\mathcal{M}_{\text{red}}| (\lambda + \delta) +  L \sum_{m_i \in \mathcal{M}_{\text{red}}} \|m_i\|^2 \Bigg]
\end{equation}
where $\gamma$ is the discount factor, $\lambda$ is the sparsity threshold, and $L$ is the smoothness constant.
\end{restatable}

See~\cref{app:proof} for the detailed proof.
The bounds presented in~\cref{thm:return_gap} rely on the assumption of bounded estimation errors, a standard analytical practice in MARL~\cite{kubatrust,chen2024rgmcomm}. To verify that these theoretical assumptions hold in practice, we tracked the empirical $\|Q^*_{\text{tot}} - Q_{\text{tot}}\|_\infty$ and the MVE loss during training, as detailed in~\cref{app:estimation_errors}.
This empirical evidence confirms that our objective enforces the practical conditions necessary to maintain these bounds.

\section{Experiments}
\label{sec:experiments}

Our experimental design is guided by the following three research questions:
\textbf{RQ1 (Communication Efficiency):} Can MUTE significantly reduce communication bandwidth without degrading the joint return?
\textbf{RQ2 (Selectivity):} Does MUTE correctly identify and prune redundant messages?
\textbf{RQ3 (Mechanism):} Are the MVE and Behavioral Anchoring components necessary for stable unlearning?

\subsection{Experimental Setup}
We implement our proposed algorithm on top of the MASIA~\cite{masia} framework, which we augment with an MLP-based message generator.
We select MASIA as our backbone because it is a value-based CTDE framework that supports unrestricted communication channels.
This structural flexibility makes it an ideal baseline for evaluating the effectiveness of our algorithm.

\textbf{Benchmarks}.
We evaluate MUTE on three benchmarks: Hallway~\cite{ndq}, StarCraft Multi-Agent Challenge (SMAC)~\cite{smac}\footnote{We use StarCraft II version SC2.4.6.2.69232.}, Traffic Junction (TJ)~\cite{ic3net}, SMAC-Communication~\cite{ndq} and SMACv2~\cite{smacv2}.
Hallway is a diagnostic grid-world environment requiring precise temporal coordination, where agents must navigate to designated positions simultaneously to succeed.
Traffic Junction~\cite{ic3net} simulates a traffic control scenario where agents travel along intersecting routes to reach destinations. It challenges agents to negotiate right-of-way and avoid collisions under limited field-of-view, making communication essential for resolving congestion.
SMAC and SMAC-Communication serve as standard testbeds for cooperative micromanagement, featuring high-dimensional combat scenarios in StarCraft II with heterogeneous units.
SMACv2 extends this by introducing stochastic start positions and random unit distributions, designed to prevent strategy memorization and test the agents' ability to generalize to unseen tactical configurations.

\textbf{Baselines.}
To validate the efficacy of MUTE, we compare it against a comprehensive suite of baselines ranging from standard value decomposition methods to advanced communication protocols.
We include QMIX~\cite{qmix} and MAPPO~\cite{mappo} as representative baselines \textit{without} explicit communication.
To evaluate the impact of restricted message exchange, we benchmark against TarMAC~\cite{tarmac}, MAIC~\cite{maic}, MASIA~\cite{masia}, IC3Net~\cite{ic3net} and SMS~\cite{sms}, which are representative MARL communication frameworks.

MUTE\footnote{Code available at \url{https://anonymous.4open.science/r/mute-2540/}} undergoes a three-stage training process: pre-training, MVE training, and communication unlearning, taking 2M, 0.5M, and 1.5M steps, respectively.
For a fair comparison, all baseline models are also trained for 4M steps using their original implementations.
This budget is more than sufficient to guarantee convergence, as most baselines were shown to converge within 2M steps in their original publications.
For evaluation, all results are reported with a 95\% confidence interval across 5 random seeds.

\begin{figure*}[t]
\centering
\includegraphics[width=\linewidth]{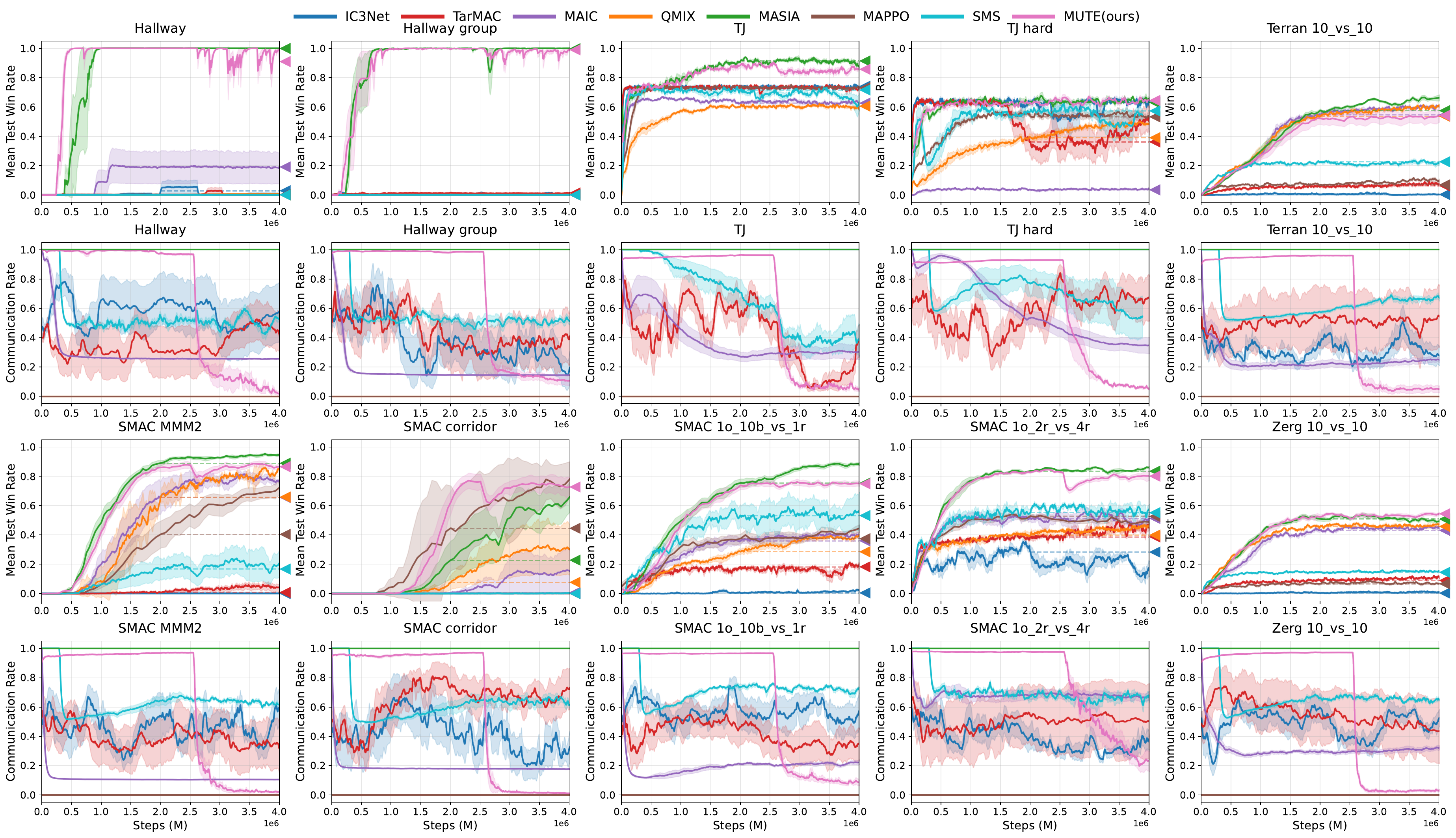}
\caption{Comparison of performance and communication efficiency against baselines across 10 benchmarks. Plots display mean test win rates and communication rates over 4 million steps. \textbf{MUTE follows a three-stage curriculum: MARL pre-training (0--2M steps), MVE training (2--2.5M steps), and communication unlearning (2.5--4M steps).} Triangles on the right axis mark the win rate at 2M steps, serving as a baseline for policy quality before the MVE and unlearning phases commence.}
\label{fig:main_results}
\vspace{-2em}
\end{figure*}

\subsection{RQ1: Communication Efficiency}
\label{sec:exp_rq1_comm}
We evaluate the communication efficiency of MUTE by comparing its converged communication rate and performance preservation against baselines across all benchmarks.
As illustrated in \Cref{fig:main_results}, MUTE consistently achieves the lowest communication overhead among all communication-enabled methods, while maintaining competitive win rates.
To rigorously quantify the performance-communication trade-off, we conduct a normalized performance-communication analysis in \cref{app:normalized_tradeoff}.
By normalizing the performance with $W_{\text{norm}} = \frac{W - W_{\text{base}}}{W_{\text{full}} - W_{\text{base}} + \epsilon}$ , we map the results onto a unified 2D plane.
The resulting scatter plots confirm that MUTE robustly dominates the communication efficiency: it is the only method that reliably achieves maximal normalized performance ($W_{\text{norm}} \ge 1.0$) while maintaining the lowest communication bandwidth.

In the Hallway and Hallway Group environments, only MUTE and MASIA achieve a 100\% win rate.
However, while MASIA requires full communication, MUTE successfully unlearns redundant messages, reducing communication by 99\% in Hallway and 90\% in Hallway Group.

In the TJ and TJ Hard scenarios, MUTE exhibits a similar pattern of efficiency.
As shown in \Cref{fig:main_results}, MUTE aggressively unlearns redundant messages, stabilizing at a communication rate of approximately 5\% while matching the optimal win rates in TJ Hard and achieving performance comparable to MASIA in TJ. %

MUTE also demonstrates superior communication efficiency in SMAC for high-dimensional combat scenarios compared to all baselines.
In the MMM2 map, MUTE matches the optimal win rates 94\% of the full-communication baseline MASIA at 2M steps, yet it has less than $<$5\% communication usage, whereas methods like TarMAC and SMS plateau at 40--60\% bandwidth.
This advantage is even more critical in the hard-exploration Corridor map.
MUTE achieves a high win rate of 73\% at 4M steps, comparable to MAPPO (77\%).
However, it is important to note that MAPPO requires the full 4M steps of MARL training to reach this level, whereas MUTE essentially freezes its policy updates after 2M steps.
As indicated by the triangles at the right side of the plot in \Cref{fig:main_results}, %
MAPPO's win rate is only 42\%  when trained with 2M steps.
More importantly, MUTE sustains this performance with near-zero ($\approx$1\%) communication, which is far more efficient compared to other communicative baselines.

Similarly, in the 1o\_10b\_vs\_1r scenario, MUTE achieves a win rate comparable to MASIA while maintaining a communication rate of 10\%.
In the 1o\_2r\_vs\_4r scenario, MUTE matches the optimal performance of the full-communication baseline MASIA but differentiates itself by reducing communication overhead to $\approx$20\%.

In the Zerg 10\_vs\_10 environment, MUTE outperforms all baselines, achieving a higher win rate 55\% than even the full-communication method MASIA 49\%, while simultaneously dropping its communication rate to near-zero (2\%) after the unlearning phase.
In the Terran 10\_vs\_10 scenario, MUTE remains competitive at the 2M step mark, comparable to the communication-heavy baselines, while utilizing less than 5\% of the available bandwidth.
\begin{figure*}[t]
\centering
\includegraphics[width=\linewidth]{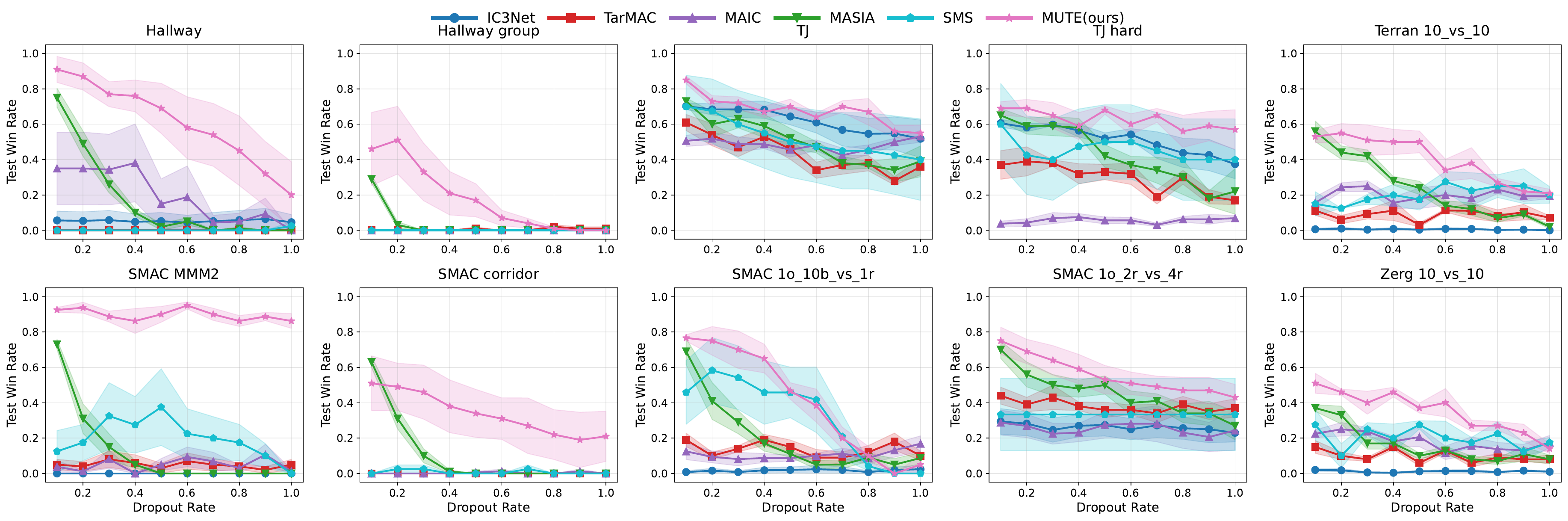}
\caption{Comparison of test win rate under varying message dropout rates against baselines across 10 benchmarks.}
\label{fig:droprate_results}
\vspace{-1em}
\end{figure*}
However, other baselines achieve higher performance at 4M steps.
We attribute this to the nature of SMACv2, which highlights stochastic start positions and random unit distributions rather than explicit coordination.
This environmental randomness allows non-communicative frameworks to achieve higher performance than in other environments, thereby diminishing the relative efficiency of communication-based methods.

\subsection{RQ2: Selectivity}
\label{sec:exp_rq2}
\begin{wraptable}{r}{0.5\linewidth}
    \centering
    \caption{AUC Summary: Area Under the Dropout-WinRate Curve. \colorbox{good}{Green} and \colorbox{bad}{Red} highlights indicate the highest and lowest values, respectively.}
    \label{tab:auc_summary}
    \resizebox{\linewidth}{!}{%
    \tablestyle{1.5pt}{1.2}
    \begin{tabular}{lcccccc}
        \toprule
        Environment & IC3Net & TarMAC & MAIC & MASIA & SMS & MUTE (ours) \\
        \midrule
        Hallway & 0.048 & \cellcolor{bad} 0.000 & 0.178 & 0.131 & 0.001 & \cellcolor{good} 0.554 \\
        Hallway Grp & 0.000 & 0.005 & \cellcolor{bad} 0.000 & 0.018 & \cellcolor{bad} 0.000 & \cellcolor{good} 0.157 \\
        10b\_vs\_1r & \cellcolor{bad} 0.015 & 0.121 & 0.093 & 0.161 & 0.294 & \cellcolor{good} 0.366 \\
        2r\_vs\_4r & 0.236 & 0.341 & \cellcolor{bad} 0.227 & 0.402 & 0.300 & \cellcolor{good} 0.498 \\
        MMM2 & \cellcolor{bad} 0.000 & 0.044 & 0.047 & 0.088 & 0.191 & \cellcolor{good} 0.808 \\
        Corridor & \cellcolor{bad} 0.000 & \cellcolor{bad} 0.000 & 0.003 & 0.074 & 0.008 & \cellcolor{good} 0.302 \\
        Terran 10v10 & \cellcolor{bad} 0.005 & 0.079 & 0.181 & 0.209 & 0.185 & \cellcolor{good} 0.364 \\
        TJ & 0.558 & \cellcolor{bad} 0.386 & 0.433 & 0.446 & 0.468 & \cellcolor{good} 0.609 \\
        TJ Hard & 0.464 & 0.269 & \cellcolor{bad} 0.051 & 0.381 & 0.405 & \cellcolor{good} 0.564 \\
        Zerg 10v10 & \cellcolor{bad} 0.011 & 0.087 & 0.162 & 0.139 & 0.178 & \cellcolor{good} 0.319 \\
        \bottomrule
    \end{tabular}%
    }
    \vspace{-1em}
\end{wraptable}
To address RQ2, we compare the win rate of MUTE across varying message dropout rates (from 10\% to 100\%) on 10 benchmarks against six communication-enabled baselines (see~\cref{fig:droprate_results}).
To quantify this robustness across the entire spectrum of communication scarcity, we calculate the Area Under the Curve (AUC) for the win rate relative to the dropout rate. \cref{tab:auc_summary} presents the summary of these AUC scores, where a higher value indicates superior resilience to information loss.

As quantified in \Cref{tab:auc_summary}, MUTE achieves the highest AUC across all 10 benchmarks, demonstrating that the reserved messages are truly performance-preserving.
In MMM2, MUTE achieves an AUC of 0.808, which is over $4\times$ higher than the best-performing baseline, SMS (0.191).
Similarly, in Hallway, MUTE's AUC (0.554) dwarfs the runner-up MAIC (0.178).
The visual trends in \Cref{fig:droprate_results} reveal distinct decay patterns between MUTE and the baselines:
Although methods like MASIA start with a similar win rate at the beginning, they exhibit a steep negative slope.
For instance, in MMM2, Hallway, Corridor and 1o\_10br\_vs\_4r, their win rates crash significantly as soon as the dropout rate exceeds 20--30\%, suggesting a high dependency on full communication.
In contrast, MUTE exhibits a gentle decay curve. Notably, in MMM2, MUTE maintains a near-optimal win rate ($\approx$90\%) even as the dropout rate approaches 80\%.

\subsection{RQ3: Ablation Studies}
\label{sec:exp_rq3}
Finally, for RQ3, we verify the necessity of our dual-objective mechanism by ablating key components.
We evaluate two variants of MUTE:
\textbf{(1) w/o Anchoring ($\beta=0$):} We minimize only the sparsity loss $\mathcal{L}_{\text{sparse}}$ without the anchoring constraint.
\textbf{(2) w/o MVE:} We replace the value-guided redundancy set $\mathcal{M}_{\text{red}}$ with a random selection of the same sparsity.

In~\cref{fig:ablation_results}, \textit{w/o Anchoring} suffers a catastrophic collapse immediately upon triggering the unlearning process. As the communication rate drops, the test win rate plummets from $\approx$95\% to below 30\%.
This confirms that $\mathcal{L}_{\text{anchor}}$ is essential for safe unlearning. Without the anchoring constraint forcing, the policy fails to internalize the coordination logic, leading to immediate failure when the communication channel is restricted.
\begin{wrapfigure}{r}{0.45\linewidth}
\centering
\includegraphics[width=\linewidth]{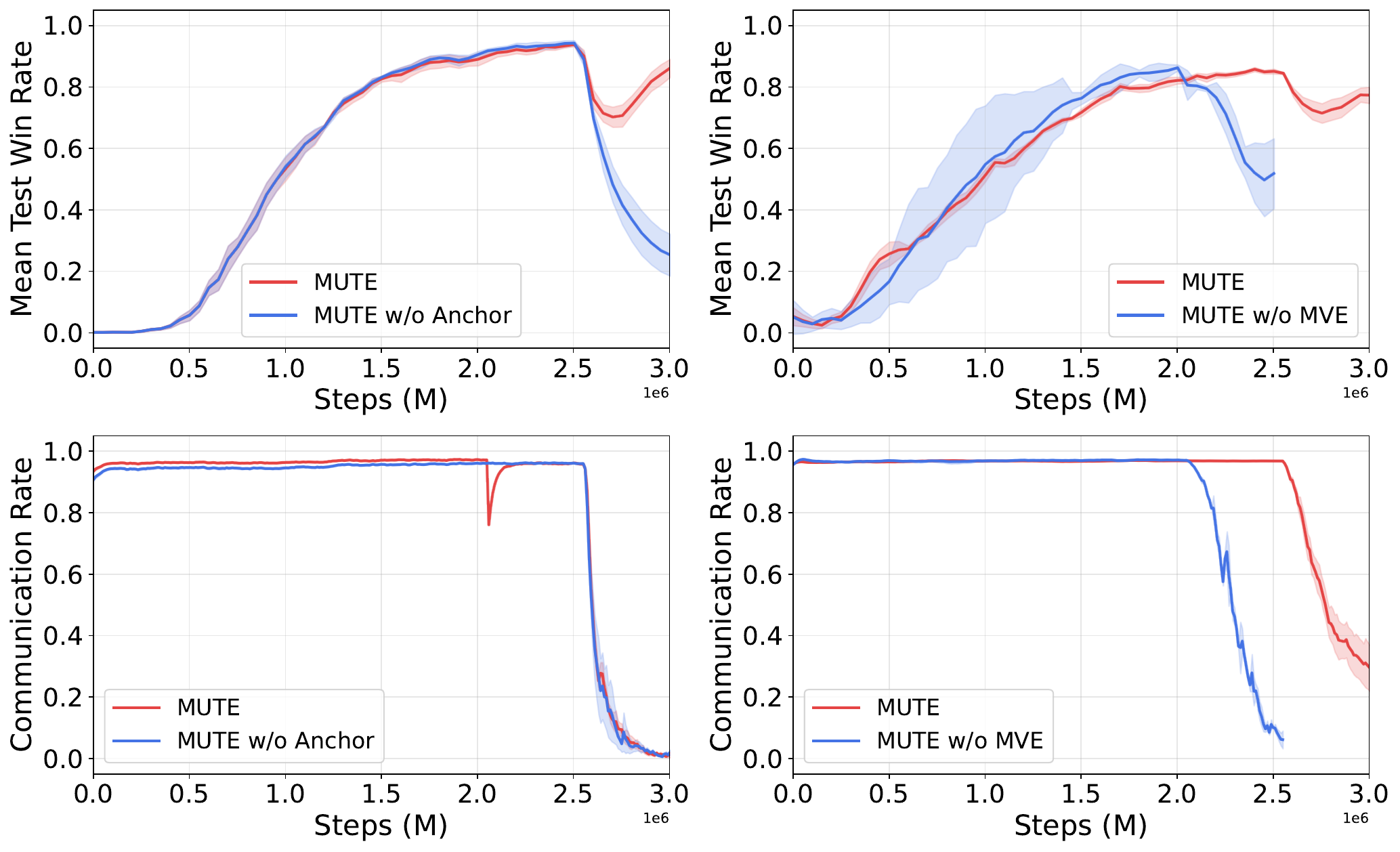}
\caption{Ablation study of \textit{w/o Anchoring} on SMAC MMM2 (left) and \textit{w/o MVE} on 1o\_10b\_vs\_1r (right).}
\label{fig:ablation_results}
\vspace{-1em}
\end{wrapfigure}
The \textit{w/o MVE} variant was pre-trained for 2M steps and 0.5M steps for unlearning; while MUTE trained for extra 0.5M steps for MVE, that causes the curve to shift by 0.5M steps.
\textit{w/o MVE} fails to maintain performance as sparsity increases. While MUTE maintains a win rate of $\approx$80\% even after the communication rate drops significantly, the win rate of \textit{w/o MVE} crashes to $\approx$50\%.
This validates the Value-Awareness of the MVE module. Randomly removing messages---even at the same sparsity level---inevitably cuts critical information. MVE ensures that only redundant messages are unlearned, preserving the critical coordination signals.

We present additional ablation experiments on hyperparameters (e.g., $\lambda$, $\beta$, and the norm constraints of $\mathcal{L}_{\text{sparse}}$) in~\cref{app:hyper_ablation}. 

Furthermore, \cref{app:budget_ablation} provides a budget-sensitivity analysis of our three-stage training schedule, confirming that our allocation was not chosen arbitrarily, but rather optimized to prevent premature unlearning under a fixed 4M-step budget.

\subsection{Analysis}

\textbf{Relation to Learned Empty Messages.}
A natural question is whether, without unlearning, standard MARL agents can simply learn to output empty messages if provided with gating mechanisms or direct communication penalties.
As we detail in \cref{app:nullmsg}, standard MARL struggles to reduce communication bandwidth purely through environmental reward feedback due to weakly attributable team rewards.
Empirical comparisons under identical training budgets demonstrate that such baselines either fail to meaningfully reduce communication or suffer severe coordination degradation, confirming the necessity of MUTE's unlearning approach.

\textbf{Computational Cost and Scalability.} A comprehensive profiling of MUTE's computational overhead---including training wall-clock time, inference FLOPs, and scalability across agent counts and hyperparameters---is provided in \cref{app:comp_profiling}.

\section{Conclusion}
We propose \textbf{MUTE}, a value-based unlearning framework. Theoretically grounded, it retains near-optimal SMAC performance even with 80--90\% communication pruning.

\textbf{Limitations.}
MUTE uses CMV as a tractable local estimate of a message's marginal contribution to the joint return. This design avoids the exponential cost of subset-based attribution methods such as Shapley value estimation, but it does not explicitly enumerate all possible message subsets. As a result, higher-order interactions among messages may not be fully captured by the MVE alone. In MUTE, this limitation is mitigated by behavioral anchoring, which provides feedback during unlearning when suppressing candidate messages would alter the behavior of the full-communication expert. Nevertheless, extending CMV to subset-aware or interaction-aware estimation could provide stronger theoretical guarantees and is an important direction for future work.

\clearpage
\section*{Acknowledgments}
This research is partially supported by the Air Force Office of Scientific Research (AFOSR), under contract FA9550-24-1-0078, and NSF award CNS-2148253. The paper was received and approved for public release by Air Force Research Laboratory (AFRL) on May 19th 2026, case number AFRL-2026-2461. Any opinions, findings, and conclusions or recommendations expressed in this material are those of the authors and do not necessarily reflect the views of AFRL or its contractors.

\bibliographystyle{unsrtnat}
\bibliography{main}

\newpage
\appendix
\onecolumn %
\crefalias{section}{appendix}
\addtocontents{toc}{\protect\setcounter{tocdepth}{2}}

\begin{center}
    \Large\bfseries Appendix Table of Contents
\end{center}
\vspace{1em}

\tableofcontents
\vspace{2em} %

\clearpage
\section{Proof of \texorpdfstring{\cref{thm:return_gap}}{Theorem \ref*{thm:return_gap}}}
\label{app:proof}
\ReturnGapThm*

\begin{proof}
We analyze the performance gap using the \textbf{Performance Difference Lemma}~\cite{kakade2002approximately}.
\small{
\begin{align}
    &J(\boldsymbol{\pi}^*) - J(\boldsymbol{\pi})
    = \frac{1}{1-\gamma} \, \mathbb{E}_{(\btau, \mathbf{a}) \sim \boldsymbol{\pi^*}} \left[ A_{\boldsymbol{\pi}}(\btau, \mathbf{a}) \right] \nonumber \\
    &= \frac{1}{1-\gamma} \, \mathbb{E}_{(\btau, \mathbf{a}) \sim \boldsymbol{\pi}^*} \left[Q_{\text{tot}}(\btau, \mathbf{a}) - V(\btau) \right] \nonumber \\
    &\leq \frac{1}{1-\gamma} \, \mathbb{E}_{(\btau, \mathbf{a}) \sim \boldsymbol{\pi}^*} \left[Q^*_{\text{tot}}(\btau, \mathbf{a}) + \epsilon - V(\btau) \right] \nonumber \\
    &= \scalebox{0.9}{$\displaystyle \frac{1}{1-\gamma} \, \mathbb{E}_{(\btau, \mathbf{a}) \sim \boldsymbol{\pi}^*} \bigg[Q^*_{\text{tot}}(\btau, \mathbf{a}) + \epsilon - \mathbb{E}_{\mathbf{a}^\prime \sim \boldsymbol{\pi}(\cdot|\btau)} \left[ Q_{\text{tot}}(\btau, \mathbf{a}^\prime)\right]\bigg]$} \nonumber \\
    &= \frac{1}{1-\gamma} \, \mathbb{E}_{\substack{(\btau, \mathbf{a}) \sim \boldsymbol{\pi}^* \\ \mathbf{a}^\prime \sim \boldsymbol{\pi}(\cdot|\btau)}} \left[Q^*_{\text{tot}}(\btau, \mathbf{m}) - Q_{\text{tot}}(\btau, \tilde{\mathbf{m}}) + \epsilon\right]
    \label{eq:gap_decomposition}
\end{align}
}
Let $\Delta \bar{Q} = Q_{\text{tot}}^*(\btau, \mathbf{m}) - Q_{\text{tot}}(\btau, \tilde{\mathbf{m}})$. We introduce an intermediate term $Q_{\text{tot}}^*(\boldsymbol{\tau}, \tilde{\mathbf{m}})$ to decompose the gap:
\begin{equation}
    \Delta \Bar{Q} = \underbrace{Q_{\text{tot}}^*(\boldsymbol{\tau}, \mathbf{m}) - Q_{\text{tot}}^*(\boldsymbol{\tau}, \tilde{\mathbf{m}})}_{\text{(A) Information Loss}} + \underbrace{Q_{\text{tot}}^*(\boldsymbol{\tau}, \tilde{\mathbf{m}}) - Q_{\text{tot}}(\boldsymbol{\tau}, \tilde{\mathbf{m}})}_{\text{(B) Approximation Error}}
\end{equation}

\textbf{Bounding Term (B) [Approximation Error]:}
This term represents the inability of $Q_{\text{tot}}$ to perfectly mimic $Q_{\text{tot}}^*$ under the sparse context, which the Anchoring Loss minimizes. Given the assumption $\|Q_{\text{tot}}^* - Q_{\text{tot}}\| \le \epsilon$, we have: $\mathbb{E}[\text{Term B}] \le \epsilon$.

\textbf{Bounding Term (A) [Information Loss]:}
This term captures the value lost by masking messages. 
Let $K = |\mathcal{M}_{\text{red}}|$ be the total number of suppressed messages. Without loss of generality, we index the messages in this subset as $\mathcal{M}_{\text{red}} = \{m_1, \dots, m_K\}$.
Term A can be exactly expanded via a telescoping sum:
\begin{equation}
\begin{split}
\text{Term (A)} &= Q_{\text{tot}}^*(\btau, \mathbf{m}) - Q_{\text{tot}}^*(\btau, \mathbf{m} \setminus \mathcal{M}_{\text{red}}) \\
&= \sum_{j=1}^K \left[ Q_{\text{tot}}^*(\btau, \mathbf{m} \setminus \mathcal{M}_{<j}) - Q_{\text{tot}}^*(\btau, \mathbf{m} \setminus \mathcal{M}_{\le j}) \right] \\
&= \sum_{j=1}^K \Delta Q_j(\btau, \mathbf{m} \setminus \mathcal{M}_{<j})
\end{split}
\end{equation}
where $\mathcal{M}_{<j} = \{m_1, \dots, m_{j-1}\}$ denotes the subset of the first $j-1$ messages removed (with $\mathcal{M}_{<1} = \emptyset$).
We can relate this to the base individual marginal drops $\Delta Q_j(\btau, \mathbf{m})$ by isolating an explicit interaction term $\mathcal{I}(\mathcal{M}_{\text{red}})$:
\begin{equation}
\text{Term (A)} = \sum_{j=1}^K \Delta Q_j(\btau, \mathbf{m}) + \mathcal{I}(\mathcal{M}_{\text{red}})
\end{equation}
where $\mathcal{I}(\mathcal{M}_{\text{red}}) = \sum_{j=1}^K \left[ \Delta Q_j(\btau, \mathbf{m} \setminus \mathcal{M}_{<j}) - \Delta Q_j(\btau, \mathbf{m}) \right]$ represents the additive information loss caused by redundant message interactions.

Since $\Delta Q_j \le \lambda + \delta$, applying~\cref{lem:interaction} yields:
\begin{equation}
\text{Term (A)} \le |\mathcal{M}_{\text{red}}|(\lambda + \delta) +  L \sum_{m_i \in \mathcal{M}_{\text{red}}} \|m_i\|^2
\end{equation}

Substituting this interaction bound, the rigorously corrected theoretical performance gap is:
\begin{equation}
J(\boldsymbol{\pi}^*) - J(\boldsymbol{\pi}) \le \frac{1}{1-\gamma} \mathbb{E} \Bigg[ 2\epsilon + |\mathcal{M}_{\text{red}}| (\lambda + \delta) +  L \sum_{m_i \in \mathcal{M}_{\text{red}}} \|m_i\|^2 \Bigg]
\end{equation}
\end{proof}

\begin{assumption} [$L$-Smoothness of the Value Function]
\label{asm:smoothness}
Let $\mathbf{m} = [m_1^\top, m_2^\top, \dots, m_N^\top]^\top \in \mathbb{R}^{N \times d_{\text{obs}}}$ be the concatenated continuous message vector of all agents. We assume the expected joint action-value function, $f(\mathbf{m}) := Q_{\text{tot}}^*(\btau, \mathbf{m})$, is continuously differentiable and $L$-smooth with respect to the continuous message embedding $\mathbf{m}$. This implies its gradient $\nabla f(\mathbf{m})$ is $L$-Lipschitz continuous. A fundamental property of $L$-smooth functions is that they can be bounded by a quadratic Taylor expansion. For any perturbation vector $\mathbf{v}$:
\begin{equation}
f(\mathbf{m} + \mathbf{v}) = f(\mathbf{m}) + \langle \nabla f(\mathbf{m}), \mathbf{v} \rangle + R(\mathbf{v})
\end{equation}
where the non-linear remainder term is strictly bounded by $|R(\mathbf{v})| \le \frac{L}{2} \|\mathbf{v}\|^2$.
\end{assumption}

\begin{lemma} [Bound on the Interaction Term]
\label{lem:interaction}
Let $\mathcal{M}_{\text{red}}$ be the set of redundant messages to be masked. The interaction term is bounded by:
\begin{equation}
\mathcal{I}(\mathcal{M}_{\text{red}}) \le L \sum_{m_i \in \mathcal{M}_{\text{red}}} \|m_i\|^2
\end{equation}
\end{lemma}

\begin{proof}
Masking a single message $m_i$ is equivalent to applying a perturbation vector $\mathbf{v}_i$, where the $i$-th block of the vector is $-m_i$ and all other blocks are $\mathbf{0}$. Note that the squared norm of this perturbation is exactly the squared norm of the message: $\|\mathbf{v}_i\|^2 = \|m_i\|^2$.    
Masking all messages in $\mathcal{M}_{\text{red}}$ simultaneously corresponds to the joint perturbation $\mathbf{v}_{\text{red}} = \sum_{m_i \in \mathcal{M}_{\text{red}}} \mathbf{v}_i$.

The interaction term is defined:
\begin{equation}
\mathcal{I}(\mathcal{M}_{\text{red}}) = \Big( Q_{\text{tot}}^*(\btau, \mathbf{m}) - Q_{\text{tot}}^*(\btau, \mathbf{m} \setminus \mathcal{M}_{\text{red}}) \Big) - \sum_{j=1}^K \Delta Q_j(\btau, \mathbf{m})
\end{equation}

Applying the Taylor expansion to both the joint perturbation $\mathbf{v}_{\text{red}}$ and the individual perturbations $\mathbf{v}_i$ yields:
\begin{equation}
\mathcal{I}(\mathcal{M}_{\text{red}}) = \Big( -\langle \nabla f(\mathbf{m}), \mathbf{v}_{\text{red}} \rangle - R(\mathbf{v}_{\text{red}}) \Big) - \sum_{m_i \in \mathcal{M}_{\text{red}}} \Big( -\langle \nabla f(\mathbf{m}), \mathbf{v}_i \rangle - R(\mathbf{v}_i) \Big)
\end{equation}
Because $\langle \nabla f(\mathbf{m}), \mathbf{v}_{\text{red}} \rangle = \sum_{m_i \in \mathcal{M}_{\text{red}}} \langle \nabla f(\mathbf{m}), \mathbf{v}_i \rangle$ by linearity,
\begin{equation}
\mathcal{I}(\mathcal{M}_{\text{red}}) = - R(\mathbf{v}_{\text{red}}) + \sum_{m_i \in \mathcal{M}_{\text{red}}} R(\mathbf{v}_i)
\end{equation}

We apply the triangle inequality to bound the absolute value of the interaction term:
\begin{equation}
|\mathcal{I}(\mathcal{M}_{\text{red}})| \le |R(\mathbf{v}_{\text{red}})| + \sum_{m_i \in \mathcal{M}_{\text{red}}} |R(\mathbf{v}_i)| \le \frac{L}{2} \|\mathbf{v}_{\text{red}}\|^2 + \sum_{m_i \in \mathcal{M}_{\text{red}}} \frac{L}{2} \|m_i\|^2
\end{equation}

Crucially, because the individual perturbations $\mathbf{v}_i$ modify disjoint, orthogonal blocks in the concatenated vector, the cross-terms in the squared norm are zero ($\langle \mathbf{v}_i, \mathbf{v}_j \rangle = 0$ for $i \neq j$). Thus, the squared norm of the joint perturbation is the sum of the squared individual norms:
\begin{equation}
\|\mathbf{v}_{\text{red}}\|^2 = \Big\| \sum_{m_i \in \mathcal{M}_{\text{red}}} \mathbf{v}_i \Big\|^2 = \sum_{m_i \in \mathcal{M}_{\text{red}}} \|\mathbf{v}_i\|^2 = \sum_{m_i \in \mathcal{M}_{\text{red}}} \|m_i\|^2
\end{equation}

Substituting this equivalence back into our bound yields the final result:
\begin{equation}
|\mathcal{I}(\mathcal{M}_{\text{red}})| \le \frac{L}{2} \sum_{m_i \in \mathcal{M}_{\text{red}}} \|m_i\|^2 + \frac{L}{2} \sum_{m_i \in \mathcal{M}_{\text{red}}} \|m_i\|^2 = L \sum_{m_i \in \mathcal{M}_{\text{red}}} \|m_i\|^2
\end{equation}
\end{proof}

\section{Complete Algorithm for MUTE}
\label{app:alg}
\begin{algorithm}[H]
    \small
   \caption{MUTE: Message Unlearning for Targeted Efficiency}
   \label{alg:mute_full}
\begin{algorithmic}[1]
   \STATE {\bfseries Input:} Pre-trained converged Policy $\boldsymbol{\pi}^*$ and Critic $Q_{\text{tot}}^*$, Replay Buffer $\mathcal{D}$.
   \STATE {\bfseries Hyperparameters:} Sparsity threshold $\lambda$, Anchor coefficient $\beta$.
   \STATE {\bfseries Initialize:} Message Value Estimator $V_\phi$, Policy $\boldsymbol{\pi}_{\theta, \psi}$ (copy from $\boldsymbol{\pi}^*$).
   \STATE \rule{\linewidth}{0.4pt}
   \STATE \textbf{Phase 1: Counterfactual Message Value Estimation}
   \FOR{iteration $\leftarrow 1$ {\bfseries to} $K_{\text{mve}}$}
       \STATE Sample batch $(\boldsymbol{\tau}, \mathbf{m}, \mathbf{a}, r, \boldsymbol{\tau}')$ from $\mathcal{D}$.
       \FOR{each agent $i$}
           \STATE \textit{// Counterfactual resampling: mask $m_i$ and resample teammates}
           \STATE $\tilde{\mathbf{a}}_{-i} \sim \boldsymbol{\pi}^*_{-i}(\cdot \mid \boldsymbol{\tau}_{-i}, \mathbf{m} \setminus \{m_i\})$.
           \STATE $y_i \leftarrow Q_{\text{tot}}^*(\boldsymbol{\tau}, \mathbf{a}) - Q_{\text{tot}}^*(\boldsymbol{\tau}, a_i, \tilde{\mathbf{a}}_{-i})$.
           \STATE $\hat{v}_i \leftarrow V_\phi(m_i)$.
       \ENDFOR
       \STATE Take a gradient step to update $\phi$ by minimizing $\mathcal{L}_{\text{MVE}} = \frac{1}{N} \sum_{i=1}^N (\hat{v}_i - \text{sg}[y_i])^2$.
   \ENDFOR

   \STATE \rule{\linewidth}{0.4pt}
   \STATE \textbf{Phase 2: Value-Guided Message Unlearning}
   \FOR{episode $\leftarrow 1$ {\bfseries to} $M_{\text{train}}$}
       \STATE Collect trajectory $\boldsymbol{\tau}$ using current policy $\boldsymbol{\pi}_{\theta, \psi}$.
       \FOR{each timestep $t$ in batch}
           \STATE Generate messages $\mathbf{m}^t = f_\psi(\boldsymbol{\tau}^t)$ and estimates $\hat{\mathbf{v}}^t = V_\phi(\mathbf{m}^t)$.
           \STATE Identify redundant set: $\mathcal{M}_{\text{red}} = \{ m_i^t \mid \hat{v}_i^t \le \lambda \}$.

           \STATE \textbf{1. Sparsity Loss:} $\mathcal{L}_{\text{sparse}} = \sum_{m \in \mathcal{M}_{\text{red}}} \| m \|_1$.
           \STATE \textbf{2. Anchor Loss:} 
           \STATE \hspace{1em} $\mathcal{L}_{\text{anchor}} = \sum_{i=1}^N \| Q_i^*(\tau_i^t, a_i^t, \mathbf{m}^t_{-i}) - Q_i(\tau_i^t, a_i^t, \mathbf{m}^t_{-i}; \theta) \|^2_2$.
           \STATE Take a gradient step to update $\theta, \psi$ by minimizing $\mathcal{L}(\theta, \psi) = \mathcal{L}_{\text{sparse}} + \beta \mathcal{L}_{\text{anchor}}$.
       \ENDFOR
   \ENDFOR
\end{algorithmic}
\end{algorithm}

\section{Environments Details}
\label{app:env} %

\begin{figure}[H]
    \centering
    \begin{subfigure}[b]{0.48\textwidth}
        \centering
        \includegraphics[height=3cm]{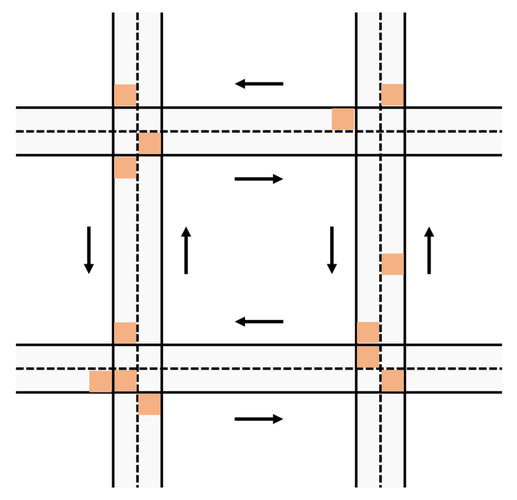}
        \caption{Traffic Junction}
    \end{subfigure}
    \hfill
    \begin{subfigure}[b]{0.48\textwidth}
        \centering
        \includegraphics[height=3cm]{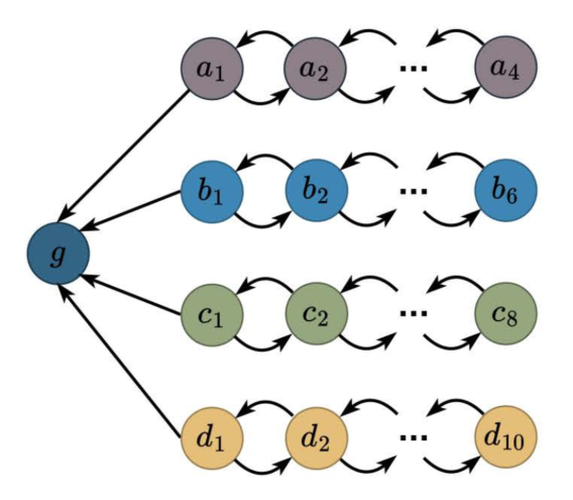}
        \caption{Hallway}
    \end{subfigure}

    \vspace{0.2cm} %

    \begin{subfigure}[b]{0.32\textwidth}
        \centering
        \includegraphics[height=2.5cm]{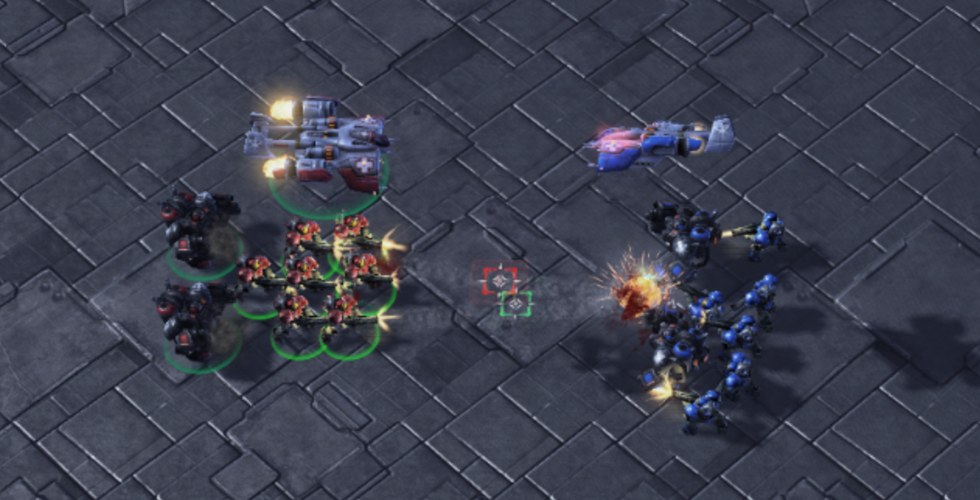}
        \caption{SMAC}
    \end{subfigure}
    \hfill
    \begin{subfigure}[b]{0.32\textwidth}
        \centering
        \includegraphics[height=2.5cm]{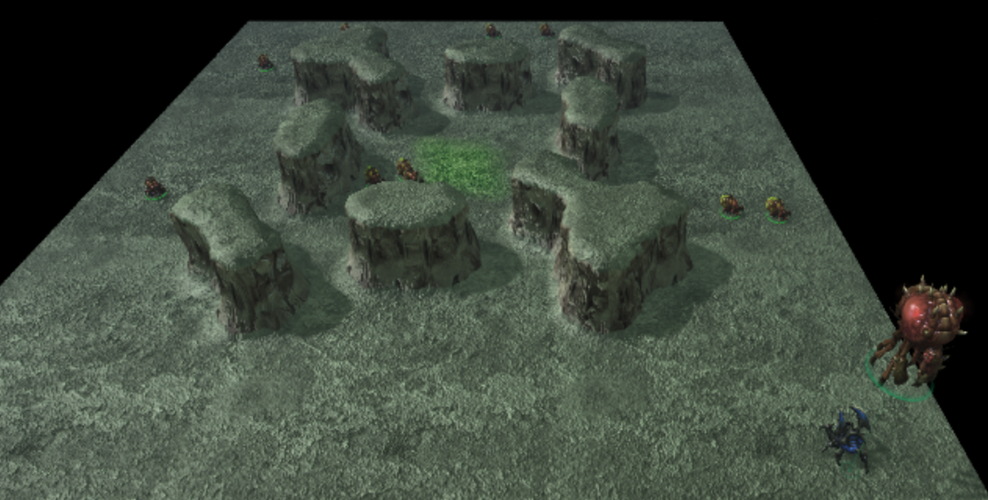}
        \caption{SMAC-Communication}
    \end{subfigure}
    \hfill
    \begin{subfigure}[b]{0.32\textwidth}
        \centering
        \includegraphics[height=2.5cm]{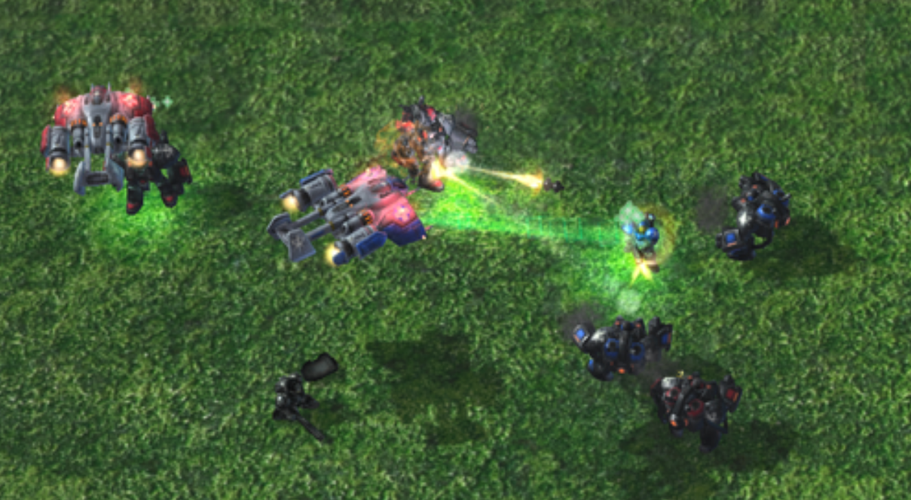}
        \caption{SMACv2}
    \end{subfigure}

    \caption{Experimental environments used in this work.}
    \label{fig:environments}
\end{figure}

We evaluate our method across four diverse cooperative multi-agent domains with varying coordination requirements and observability constraints.

\subsection{StarCraft Multi-Agent Challenge (SMAC)}

SMAC~\cite{smac} and SMAC-Communication~\cite{ndq} provides a set of cooperative micromanagement scenarios in StarCraft II where agents must coordinate to defeat enemy units controlled by the built-in game AI. Each agent controls a single unit with local observations limited to a sight range of 9 units, including ally and enemy health, positions, and unit types within view. Agents can move in four cardinal directions, stop, or attack enemies within range (shoot range of 6 units). The team receives a shared reward based on damage dealt and eliminating enemies, with a bonus reward of 200 for winning the battle.

We evaluate on the following scenarios:
\begin{itemize}
    \item \textbf{MMM2}: A heterogeneous battle featuring 1 Medivac, 2 Marauders, and 7 Marines against the same enemy composition. The Medivac provides healing support but cannot attack, requiring agents to learn to protect it while coordinating offensive actions. This scenario demands role-based coordination between healers and combat units.

    \item \textbf{Corridor}: A super-hard scenario where 6 Zealots must defeat 24 Zerglings in a narrow corridor. The confined space creates complex movement dependencies, and agents must coordinate their positions to avoid blocking each other while maximizing damage output.
    
    \item \textbf{1o\_10b\_vs\_1r}: An asymmetric micromanagement task where 1 Overseer (detector unit without attack capability) and 10 Banelings (suicide bomber units) must defeat a single Roach. Banelings deal splash damage upon death, requiring precise timing and positioning to maximize damage while avoiding friendly fire.
    
    \item \textbf{1o\_2r\_vs\_4r}: 1 Overseer and 2 Roaches versus 4 enemy Roaches. The numerical disadvantage requires optimal focus-fire strategies and kiting maneuvers to overcome the enemy's superior firepower.

\end{itemize}

\subsection{SMACv2}

SMACv2~\cite{smacv2} extends SMAC with procedurally generated scenarios featuring randomized unit compositions, starting positions, and team configurations. Unlike the fixed scenarios in SMAC, SMACv2 evaluates generalization capabilities by randomizing both ally and enemy team compositions at the start of each episode. Units are sampled from weighted distributions, and starting positions follow a ``surrounded and reflect'' distribution where teams may start in encircled configurations.

We evaluate on three faction-specific scenarios:
\begin{itemize}
    \item \textbf{Terran 10v10}: 10 allied Terran units versus 10 enemy units, with unit types sampled from Marines (45\%), Marauders (45\%), and Medivacs (10\%). Marines provide ranged DPS, Marauders offer tankier ranged support with bonus damage against armored units, and Medivacs heal nearby allies.
    
    \item \textbf{Zerg 10v10}: 10 allied Zerg units versus 10 enemies, sampled from Zerglings (45\%), Hydralisks (45\%), and Banelings (10\%). Zerglings are fast melee units, Hydralisks provide ranged attacks, and Banelings are explosive suicide units requiring careful positioning.
\end{itemize}

\subsection{Hallway}

The Hallway environment~\cite{ndq} is a coordination benchmark where $n$ agents navigate one-dimensional corridors of varying lengths. Each agent $i$ occupies a discrete position in a hallway of length $L_i$ and can choose from three actions: stay, move left (decrease position), or move right (increase position). The team receives a reward only when \textit{all} agents simultaneously reach position 0. If any agent reaches position 0 while others have not, the episode terminates with zero reward. Each agent observes only its own position, requiring implicit coordination through communication.

\begin{itemize}
    \item \textbf{Hallway (4 agents)}: Four agents with corridor lengths $[4, 6, 8, 10]$. Agents must learn to synchronize their movements despite different travel distances, with faster agents learning to wait for slower ones. The episode limit is set to $\max(L_i) + 10 = 20$ steps.
    
    \item \textbf{Hallway-Group (7 agents)}: Seven agents divided into two groups: Group 0 (3 agents, corridor lengths $[3, 5, 7]$) and Group 1 (4 agents, corridor lengths $[4, 6, 8, 10]$). Groups must reach position 0 \textit{sequentially}---if both groups complete simultaneously, a penalty is incurred and the episode continues. This creates a hierarchical coordination problem where agents must first coordinate within their group, then coordinate the timing between groups.
\end{itemize}

\subsection{Traffic Junction}

Traffic Junction~\cite{ic3net} simulates a multi-agent traffic coordination problem where autonomous vehicles must navigate through intersections without collisions. Cars spawn probabilistically at entry points and follow predetermined routes through the junction. Each car can either accelerate (GAS: move forward one cell) or brake (BRAKE: stay in place). Agents receive a timestep penalty of $-0.01$ per step while active and a collision penalty of $-10$ when two cars occupy the same cell. The episode succeeds when all cars reach their destinations without collisions.

\begin{itemize}
    \item \textbf{Traffic Junction (Medium)}: A $14 \times 14$ grid with a 4-way intersection and 10 cars. Cars spawn with probability in range $[0.05, 0.2]$ and must navigate through the central junction. Each car has zero vision (observes only its own state), requiring communication to avoid collisions. The episode limit is 40 timesteps, and success requires learning implicit right-of-way protocols.
    
    \item \textbf{Traffic Junction (Hard)}: An $18 \times 18$ grid featuring a more complex road network with 8 possible routes and 20 cars. The spawning rate is lower ($[0.02, 0.05]$) but the increased number of agents and longer routes (episode limit of 80 timesteps) create more challenging coordination requirements. The grid contains multiple intersections, requiring agents to handle cascading traffic flow decisions.
\end{itemize}

\section{Model Details}
We adopt the same backbone architecture as MASIA~\cite{masia}.
The primary distinction is the inclusion of a message generator $f_{\psi}$ at the sender side.
This generator is parameterized as a single Linear layer mapping from $d_{\text{obs}}$ to $d_{\text{obs}}$ with no non-linear activation.

\section{Training Details}

We detail the specific hyperparameters used for the Message Value Estimator (MVE) training and the subsequent unlearning phase in Table~\ref{tab:combined_hyperparams}.
Across both phases, network parameters are updated using the Adam optimizer.
For the MVE (Phase 2), we utilize a multi-head attention mechanism ($H=4$) to capture complex agent interactions, training for $K_{\text{mve}}$ steps to ensure accurate counterfactual value estimation.
In the unlearning phase (Phase 3), the sparsity threshold $\lambda$ controls the aggressiveness of message pruning, while the anchoring coefficient $\beta$ stabilizes the policy by penalizing deviations from the teacher's original $Q$-values.

\begin{table}[h]
\centering
\caption{Training Hyperparameters for Phase 2 (MVE) and Phase 3 (Unlearning)}
\label{tab:combined_hyperparams}
\begin{tabular}{@{}llc|llc@{}}
\toprule
\multicolumn{3}{c|}{\textbf{Phase 2: Context-Aware MVE}} & \multicolumn{3}{c}{\textbf{Phase 3: Value-Conditional Unlearning}} \\
\midrule
\textbf{Sym.} & \textbf{Parameter} & \textbf{Value} & \textbf{Sym.} & \textbf{Parameter} & \textbf{Value} \\
\midrule
$K_{\text{mve}}$ & Training steps & 500,000 & $M_{\text{unlearn}}$ & Training steps & 1,500,000 \\
$d_{\text{mve}}$ & MVE hidden dim & 64 & $\lambda$ & Sparsity threshold & 1.0 \\
$H$ & Attention heads & 4 & $\beta$ & Anchoring weight & 10 \\
$\eta_{\text{mve}}$ & MVE Learning rate & 0.0005 & $\tau_{\text{ema}}$ & EMA momentum & 0.01 \\
\bottomrule
\end{tabular}
\end{table}

\section{Resolving Duplicate Messages Through Joint Optimization}
\label{app:duplicate_message}
A subtle case with the CMV in \cref{eq:cmv} is that two messages carrying identical information would each receive a marginal value of zero: masking either alone causes no return loss, since the other still carries the same information. A naive thresholding rule would then place both in $\mathcal{M}_{\text{red}}$ and suppress them jointly, eliminating the information entirely.
We argue that the joint objective $\mathcal{L}$ optimizes toward keeping exactly one of the two messages, without jointly suppressing useful information.

Consider two messages $m_i$ and $m_k$ carrying identical information, and examine three configurations of the unlearning dynamics. (i) Both messages active: $\mathcal{L}_{\text{anchor}}$ is small but $\mathcal{L}_{\text{sparse}}$ is large, since both fall below $\lambda$ and are penalized by \cref{eq:sparsity_loss}. (ii) Both messages suppressed: $\mathcal{L}_{\text{sparse}}$ vanishes, but the shared information is lost, the individual Q-values diverge from the frozen $Q_i^*$, and $\mathcal{L}_{\text{anchor}}$ grows as a result. (iii) Exactly one message active: $\mathcal{L}_{\text{anchor}}$ remains small because the information is still transmitted through the surviving channel, and $\mathcal{L}_{\text{sparse}}$ is reduced relative to (i) because only one message is being penalized. Configuration (iii) is therefore the minimizer of $\mathcal{L}$ over the three.

The unlearning phase of MUTE naturally navigates the network toward this minimizer.
Driven by the opposing continuous gradients—where the sparsity loss pushes away from keeping both messages, and the behavioral anchoring loss penalizes the performance degradation of suppressing both—the optimization process dynamically adjusts the message generator and receiver policies. This continuous adjustment naturally breaks the initial symmetry between the identical messages, allowing the system to smoothly converge to configuration (iii) and eliminate the redundant message without sacrificing task performance.

\section{Ablation Study: State-Aware Message Value Estimator}
\label{app:mve_ablation}

\begin{table}[H]
\centering
\small
\caption{Comparison of Win Rate and Communication Rate between the original message-only MUTE and an ablation variant using both state and messages.}
\label{tab:mve_ablation}
\resizebox{0.8\linewidth}{!}{
\tablestyle{1.5pt}{1.2}
\begin{tabular}{lcccc}
\toprule
 & \multicolumn{2}{c}{\textbf{Win Rate}} & \multicolumn{2}{c}{\textbf{Comm Rate}} \\
\cmidrule(lr){2-3} \cmidrule(lr){4-5}
\textbf{Environment} & \textbf{MUTE (ours)} & \textbf{MUTE state+msg} & \textbf{MUTE (ours)} & \textbf{MUTE state+msg} \\
\midrule
\texttt{smac\_1o\_10b\_vs\_1r} & 0.77 $\pm$ 0.02 & 0.75 $\pm$ 0.04 & 0.06 $\pm$ 0.01 & 0.05 $\pm$ 0.01 \\
\texttt{hallway}               & 0.97 $\pm$ 0.09 & 0.96 $\pm$ 0.12 & 0.02 $\pm$ 0.02 & 0.14 $\pm$ 0.06 \\
\texttt{smac\_1o\_2r\_vs\_4r}  & 0.81 $\pm$ 0.03 & 0.80 $\pm$ 0.02 & 0.23 $\pm$ 0.07 & 0.34 $\pm$ 0.25 \\
\bottomrule
\end{tabular}
}
\end{table}

\begin{figure}[H]
\centering
\includegraphics[width=0.9\linewidth]{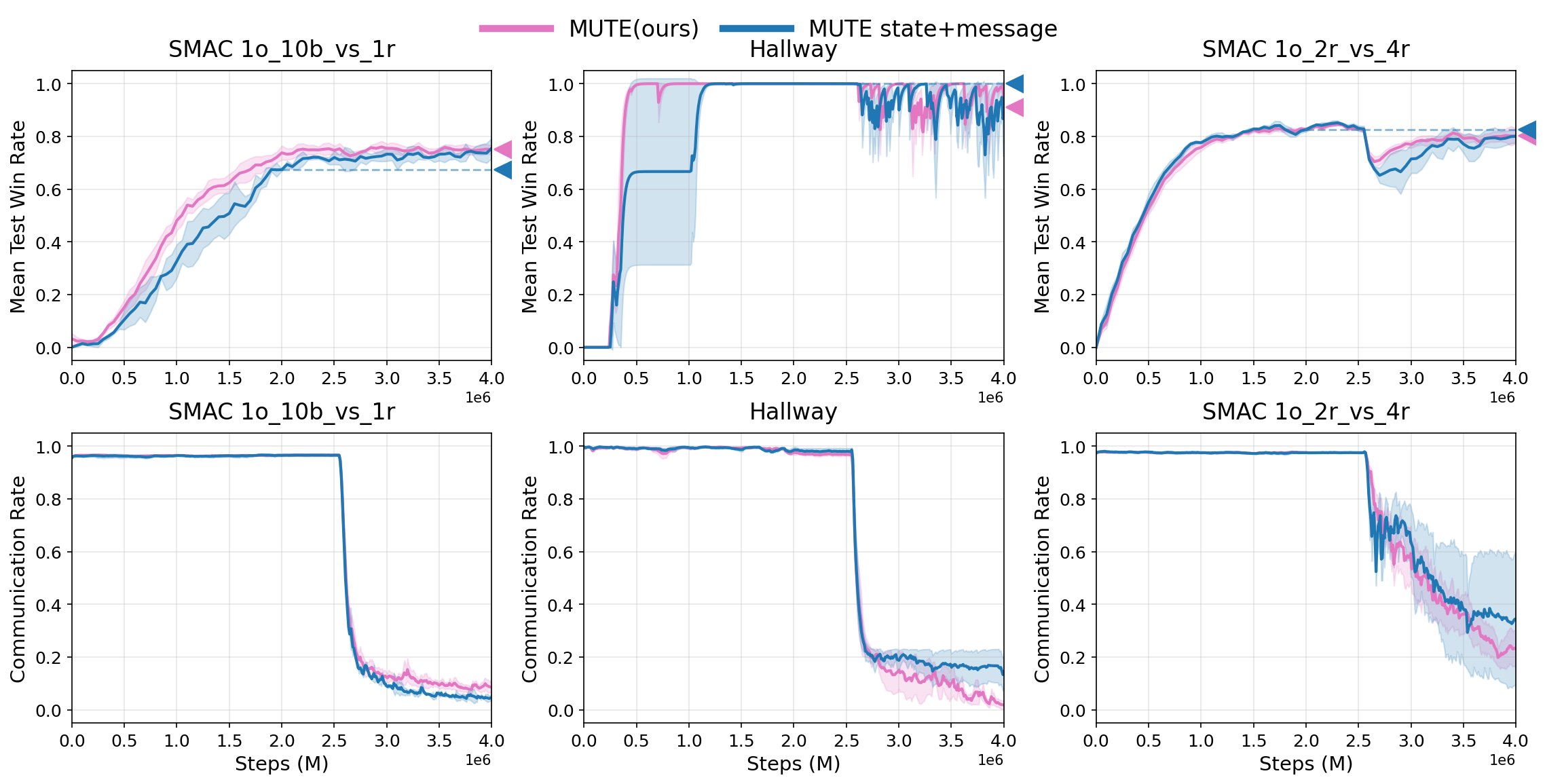}
\caption{Training curves of Win Rate and Communication Rate between the original message-only MUTE and an ablation variant using both state and messages.}
\label{fig:mve_ablation}
\vspace{-1em}
\end{figure}

In the main text, we design the Message Value Estimator (MVE) to evaluate the marginal contribution of a message using only the joint set of messages, $V_\phi(\mathbf{m})$.
While message importance may depend on the state in addition to the other agents' messages, we hypothesize that the joint message set inherently encapsulates sufficient environmental context because each message is generated from an agent's local observation history. 

To directly test this and validate our design choice, we implemented an additional variant of MUTE in which the MVE takes both the global state and the joint messages as input, denoted as $V_\phi(s, \mathbf{m})$. We compared this state-aware variant against our original message-only MVE across three environments. The results for both task performance (Win Rate) and communication bandwidth (Comm Rate) are presented in ~\cref{tab:mve_ablation} and~\cref{fig:mve_ablation}.

As shown in the table and the figure, the two variants exhibit very similar task performance across all evaluated environments. Furthermore, the original message-only MVE generally achieves a better communication bandwidth trade-off. Specifically, on \texttt{hallway}, the original MVE learns substantially faster and reaches a much lower final communication rate. Similarly, on \texttt{1o\_2r\_vs\_4r}, the \texttt{state$+$message} variant requires clearly higher communication after unlearning. While the state-aware variant achieves a slightly lower communication rate on \texttt{1o\_10b\_vs\_1r}, the overall results demonstrate that adding explicit state information provides negligible benefits while often degrading unlearning efficiency. This empirical evidence strongly supports our choice of a message-only MVE formulation.

\section{Hyperparameter Ablation Study}
\label{app:hyper_ablation}

In this section, we provide the detailed hyperparameter ablation study of MUTE across 10 benchmarks.

\begin{table}[h]
    \centering
    \caption{Test Win Rate of MUTE using different values of $\lambda$ across 10 benchmarks.}
    \label{tab:lambda_ablation_win}
    \resizebox{\linewidth}{!}{%
    \tablestyle{1.5pt}{1.2}
    \begin{tabular}{lcccccccccc}
        \toprule
        Method & Hallway & Hallway Grp & TJ & TJ Hard & Terran 10\_vs\_10 & MMM2 & Corridor & 1o\_10b\_vs\_1r & 1o\_2r\_vs\_4r & Zerg 10\_vs\_10 \\
        \midrule
        MUTE & \cellcolor{good}$1.00 \pm 0.10$ & \cellcolor{good}$1.00 \pm 0.08$ & \cellcolor{good}$0.84 \pm 0.03$ & \cellcolor{good}$0.65 \pm 0.06$ & \cellcolor{good}$0.57 \pm 0.07$ & $0.86 \pm 0.05$ & \cellcolor{good}$0.72 \pm 0.27$ & $0.74 \pm 0.04$ & \cellcolor{good}$0.81 \pm 0.09$ & \cellcolor{good}$0.55 \pm 0.03$ \\
        $\lambda=0.1\cdot \mu_{MVE}$ & $0.65 \pm 0.31$ & $0.92 \pm 0.06$ & $0.83 \pm 0.05$ & $0.64 \pm 0.04$ & $0.51 \pm 0.07$ & $0.84 \pm 0.03$ & $0.54 \pm 0.23$ & \cellcolor{good}$0.76 \pm 0.02$ & $0.73 \pm 0.11$ & $0.52 \pm 0.03$ \\
        $\lambda=0.5\cdot \mu_{MVE}$ & $0.62 \pm 0.35$ & $0.86 \pm 0.13$ & $0.81 \pm 0.06$ & $0.61 \pm 0.04$ & $0.51 \pm 0.10$ & \cellcolor{good}$0.87 \pm 0.05$ & $0.50 \pm 0.22$ & $0.70 \pm 0.05$ & $0.64 \pm 0.10$ & \cellcolor{bad}$0.50 \pm 0.04$ \\
        $\lambda=2.0\cdot \mu_{MVE}$ & \cellcolor{bad}$0.55 \pm 0.35$ & \cellcolor{bad}$0.59 \pm 0.10$ & \cellcolor{bad}$0.76 \pm 0.07$ & \cellcolor{bad}$0.59 \pm 0.05$ & \cellcolor{bad}$0.48 \pm 0.07$ & \cellcolor{bad}$0.80 \pm 0.03$ & \cellcolor{bad}$0.49 \pm 0.22$ & \cellcolor{bad}$0.39 \pm 0.15$ & \cellcolor{bad}$0.57 \pm 0.11$ & $0.54 \pm 0.03$ \\
        \bottomrule
    \end{tabular}%
    }
    \vspace{-1em}
\end{table}

\begin{table}[h]
    \centering
    \caption{Communication Rate of MUTE using different values of $\lambda$ across 10 benchmarks.}
    \label{tab:lambda_ablation_comm}
    \resizebox{\linewidth}{!}{%
    \tablestyle{1.5pt}{1.2}
    \begin{tabular}{lcccccccccc}
        \toprule
        Method & Hallway & Hallway Grp & TJ & TJ Hard & Terran 10\_vs\_10 & MMM2 & Corridor & 1o\_10b\_vs\_1r & 1o\_2r\_vs\_4r & Zerg 10\_vs\_10 \\
        \midrule
        MUTE & \cellcolor{good}$0.03 \pm 0.06$ & \cellcolor{good}$0.10 \pm 0.10$ & $0.02 \pm 0.01$ & \cellcolor{good}$0.05 \pm 0.02$ & $0.05 \pm 0.07$ & $0.02 \pm 0.01$ & \cellcolor{good}$0.01 \pm 0.02$ & $0.08 \pm 0.05$ & $0.23 \pm 0.20$ & $0.03 \pm 0.02$ \\
        $\lambda=0.1\cdot \mu_{MVE}$ & \cellcolor{bad}$0.38 \pm 0.28$ & \cellcolor{bad}$0.15 \pm 0.08$ & \cellcolor{bad}$0.09 \pm 0.02$ & \cellcolor{bad}$0.30 \pm 0.26$ & \cellcolor{bad}$0.28 \pm 0.24$ & \cellcolor{bad}$0.18 \pm 0.19$ & \cellcolor{bad}$0.20 \pm 0.24$ & \cellcolor{bad}$0.20 \pm 0.20$ & \cellcolor{bad}$0.24 \pm 0.16$ & \cellcolor{bad}$0.15 \pm 0.10$ \\
        $\lambda=0.5\cdot \mu_{MVE}$ & $0.26 \pm 0.27$ & $0.14 \pm 0.04$ & $0.02 \pm 0.03$ & $0.14 \pm 0.08$ & $0.19 \pm 0.23$ & $0.08 \pm 0.11$ & $0.03 \pm 0.04$ & $0.07 \pm 0.07$ & $0.11 \pm 0.09$ & $0.02 \pm 0.02$ \\
        $\lambda=2.0\cdot \mu_{MVE}$ & $0.16 \pm 0.18$ & $0.13 \pm 0.09$ & \cellcolor{good}$0.01 \pm 0.05$ & $0.05 \pm 0.05$ & \cellcolor{good}$0.00 \pm 0.00$ & \cellcolor{good}$0.01 \pm 0.01$ & $0.02 \pm 0.03$ & \cellcolor{good}$0.00 \pm 0.00$ & \cellcolor{good}$0.07 \pm 0.10$ & \cellcolor{good}$0.00 \pm 0.00$ \\
        \bottomrule
    \end{tabular}%
    }
    \vspace{-1em}
\end{table}

\textbf{Impact of $\lambda$ in $\mathcal{M}_{\text{red}}$.} We evaluate the effect of the threshold multiplier $\lambda$ by comparing values of 0.1, 0.5, and 2.0 against the MUTE default of 1.0. A larger $\lambda$ imposes a stricter penalty on communication, whereas a smaller $\lambda$ results in a less restrictive bandwidth when selecting low-value message candidates for unlearning. Here, the threshold is scaled by $\mu_{\text{MVE}}$, which represents the moving average of the MVE predictions during training (i.e., the expected message value). As shown in~\cref{tab:lambda_ablation_win} and~\cref{tab:lambda_ablation_comm}, utilizing a relaxed threshold ($\lambda=0.1\cdot \mu_{\text{MVE}}$ or $0.5\cdot \mu_{\text{MVE}}$) yields the highest communication rates across the benchmarks. Conversely, an overly strict threshold ($\lambda=2.0\cdot \mu_{\text{MVE}}$) aggressively reduces the communication rate to the lowest level---even reaching 0.0 in three environments---but this severely degrades task performance. The default setting ($\lambda=1.0$) strikes the optimal balance, successfully minimizing communication bandwidth while preserving the highest overall task performance across the benchmarks.

\begin{table}[H]
    \centering
    \caption{Test Win Rate of MUTE using different values of $\beta$ across 10 benchmarks.}
    \label{tab:beta_ablation_win}
    \resizebox{\linewidth}{!}{%
    \tablestyle{1.5pt}{1.2}
    \begin{tabular}{lcccccccccc}
        \toprule
        Method & Hallway & Hallway Grp & TJ & TJ Hard & Terran 10\_vs\_10 & MMM2 & Corridor & 1o\_10b\_vs\_1r & 1o\_2r\_vs\_4r & Zerg 10\_vs\_10 \\
        \midrule
        MUTE & \cellcolor{good}$1.00 \pm 0.10$ & \cellcolor{good}$1.00 \pm 0.08$ & $0.84 \pm 0.03$ & \cellcolor{good}$0.65 \pm 0.06$ & \cellcolor{good}$0.57 \pm 0.07$ & $0.86 \pm 0.05$ & \cellcolor{good}$0.72 \pm 0.27$ & $0.74 \pm 0.04$ & \cellcolor{good}$0.81 \pm 0.09$ & \cellcolor{good}$0.55 \pm 0.03$ \\
        $\beta=1$ & \cellcolor{bad}$0.49 \pm 0.34$ & \cellcolor{bad}$0.97 \pm 0.03$ & \cellcolor{bad}$0.81 \pm 0.03$ & $0.64 \pm 0.06$ & \cellcolor{bad}$0.53 \pm 0.14$ & $0.86 \pm 0.04$ & \cellcolor{bad}$0.57 \pm 0.24$ & \cellcolor{bad}$0.73 \pm 0.03$ & \cellcolor{bad}$0.78 \pm 0.03$ & \cellcolor{bad}$0.51 \pm 0.05$ \\
        $\beta=50$ & $0.80 \pm 0.31$ & \cellcolor{good}$1.00 \pm 0.00$ & \cellcolor{good}$0.85 \pm 0.04$ & \cellcolor{bad}$0.53 \pm 0.07$ & $0.54 \pm 0.09$ & \cellcolor{good}$0.88 \pm 0.04$ & $0.58 \pm 0.23$ & \cellcolor{good}$0.79 \pm 0.06$ & \cellcolor{good}$0.81 \pm 0.04$ & $0.52 \pm 0.03$ \\
        \bottomrule
    \end{tabular}%
    }
    \vspace{-1em}
\end{table}

\begin{table}[H]
    \centering
    \caption{Communication Rate of MUTE using different values of $\beta$ across 10 benchmarks.}
    \label{tab:beta_ablation_comm}
    \resizebox{\linewidth}{!}{
    \tablestyle{1.5pt}{1.2}
    \begin{tabular}{lcccccccccc}
        \toprule
        Method & Hallway & Hallway Grp & TJ & TJ Hard & Terran 10\_vs\_10 & MMM2 & Corridor & 1o\_10b\_vs\_1r & 1o\_2r\_vs\_4r & Zerg 10\_vs\_10 \\
        \midrule
        MUTE & \cellcolor{good}$0.03 \pm 0.06$ & \cellcolor{good}$0.10 \pm 0.10$ & \cellcolor{good}$0.02 \pm 0.01$ & $0.05 \pm 0.02$ & $0.05 \pm 0.07$ & \cellcolor{good}$0.02 \pm 0.01$ & \cellcolor{good}$0.01 \pm 0.02$ & $0.08 \pm 0.05$ & \cellcolor{bad}$0.23 \pm 0.20$ & \cellcolor{good}$0.03 \pm 0.02$ \\
        $\beta=1$ & $0.18 \pm 0.08$ & $0.11 \pm 0.06$ & $0.03 \pm 0.01$ & \cellcolor{good}$0.04 \pm 0.04$ & \cellcolor{good}$0.01 \pm 0.01$ & \cellcolor{good}$0.02 \pm 0.00$ & $0.02 \pm 0.03$ & \cellcolor{good}$0.02 \pm 0.01$ & \cellcolor{good}$0.06 \pm 0.02$ & \cellcolor{good}$0.03 \pm 0.02$ \\
        $\beta=50$ & \cellcolor{bad}$0.30 \pm 0.18$ & \cellcolor{bad}$0.18 \pm 0.06$ & \cellcolor{bad}$0.07 \pm 0.04$ & \cellcolor{bad}$0.10 \pm 0.04$ & \cellcolor{bad}$0.06 \pm 0.04$ & \cellcolor{bad}$0.06 \pm 0.04$ & \cellcolor{bad}$0.04 \pm 0.03$ & \cellcolor{bad}$0.15 \pm 0.03$ & $0.22 \pm 0.06$ & \cellcolor{bad}$0.09 \pm 0.01$ \\
        \bottomrule
    \end{tabular}%
    }
    \vspace{-1em}
\end{table}

\textbf{Sensitivity to $\beta$.}
We ablate the hyperparameter $\beta$, which controls the weight of the behavioral anchoring during the unlearning phase, comparing the MUTE default of $\beta=10$ against extreme values of 1 and 50.
As shown in~\cref{tab:beta_ablation_win} and~\cref{tab:beta_ablation_comm}, setting $\beta$ to a small value ($\beta=1$) results in a severe drop in task performance (Win Rate) across the majority of environments. This indicates that when $\beta$ is too weak, the model suffers from catastrophic forgetting, failing to preserve essential cooperative behaviors as messages are aggressively removed. Conversely, setting $\beta$ to an overly large value ($\beta=50$) makes the policy too rigid. While it recovers much of the task performance compared to $\beta=1$, it resists unlearning, yielding the worst (highest) communication rates across nearly all benchmarks. The default setting ($\beta=10$) successfully navigates this trade-off, providing an optimal level of anchoring to preserve high win rates while still allowing the model to safely and efficiently reduce communication bandwidth.

\begin{table}[H]
    \centering
    \caption{Test Win Rate of MUTE with $L_1$ and $L_2$ norms for $\mathcal{L}_{\text{sparse}}$ across 10 benchmarks.}
    \label{tab:norm_ablation_win}
    \resizebox{\linewidth}{!}{%
    \tablestyle{1.5pt}{1.2}
    \begin{tabular}{lcccccccccc}
        \toprule
        Method & Hallway & Hallway Grp & TJ & TJ Hard & Terran 10\_vs\_10 & MMM2 & Corridor & 1o\_10b\_vs\_1r & 1o\_2r\_vs\_4r & Zerg 10\_vs\_10 \\
        \midrule
        MUTE & \cellcolor{good}$1.00 \pm 0.10$ & \cellcolor{good}$1.00 \pm 0.08$ & \cellcolor{bad}$0.84 \pm 0.03$ & \cellcolor{good}$0.65 \pm 0.06$ & \cellcolor{good}$0.57 \pm 0.07$ & \cellcolor{good}$0.86 \pm 0.05$ & \cellcolor{good}$0.72 \pm 0.27$ & \cellcolor{good}$0.74 \pm 0.04$ & \cellcolor{good}$0.81 \pm 0.09$ & \cellcolor{good}$0.55 \pm 0.03$ \\
        $L_2$ norm & \cellcolor{bad}$0.54 \pm 0.35$ & \cellcolor{bad}$0.80 \pm 0.31$ & \cellcolor{good}$0.85 \pm 0.03$ & \cellcolor{bad}$0.57 \pm 0.06$ & \cellcolor{bad}$0.54 \pm 0.07$ & \cellcolor{bad}$0.82 \pm 0.08$ & \cellcolor{bad}$0.56 \pm 0.23$ & \cellcolor{bad}$0.65 \pm 0.10$ & \cellcolor{bad}$0.68 \pm 0.11$ & \cellcolor{bad}$0.52 \pm 0.04$ \\
        \bottomrule
    \end{tabular}%
    }
    \vspace{-1em}
\end{table}

\begin{table}[H]
    \centering
    \caption{Communication Rate of MUTE with $L_1$ and $L_2$ norms for $\mathcal{L}_{\text{sparse}}$ across 10 benchmarks.}
    \label{tab:norm_ablation_comm}
    \resizebox{\linewidth}{!}{%
    \tablestyle{1.5pt}{1.2}
    \begin{tabular}{lcccccccccc}
        \toprule
        Method & Hallway & Hallway Grp & TJ & TJ Hard & Terran 10\_vs\_10 & MMM2 & Corridor & 1o\_10b\_vs\_1r & 1o\_2r\_vs\_4r & Zerg 10\_vs\_10 \\
        \midrule
        MUTE & \cellcolor{good}$0.03 \pm 0.06$ & \cellcolor{good}$0.10 \pm 0.10$ & \cellcolor{good}$0.02 \pm 0.01$ & \cellcolor{good}$0.05 \pm 0.02$ & \cellcolor{good}$0.05 \pm 0.07$ & \cellcolor{good}$0.02 \pm 0.01$ & \cellcolor{good}$0.01 \pm 0.02$ & \cellcolor{bad}$0.08 \pm 0.05$ & \cellcolor{bad}$0.23 \pm 0.20$ & \cellcolor{bad}$0.03 \pm 0.02$ \\
        $L_2$ norm & \cellcolor{bad}$0.31 \pm 0.15$ & \cellcolor{bad}$0.13 \pm 0.08$ & \cellcolor{bad}$0.04 \pm 0.03$ & \cellcolor{good}$0.05 \pm 0.03$ & \cellcolor{good}$0.05 \pm 0.06$ & \cellcolor{bad}$0.03 \pm 0.00$ & \cellcolor{bad}$0.02 \pm 0.01$ & \cellcolor{good}$0.03 \pm 0.02$ & \cellcolor{good}$0.22 \pm 0.25$ & \cellcolor{good}$0.01 \pm 0.01$ \\
        \bottomrule
    \end{tabular}%
    }
    \vspace{-1em}
\end{table}

\textbf{Impact of the Norm Constraint in $\mathcal{L}_{\text{sparse}}$.}
We ablate the choice of regularization norm used in the sparsity objective, comparing MUTE's default $L_1$  against an $L_2$ norm. As shown in~\cref{tab:norm_ablation_win}  and~\cref{tab:norm_ablation_comm} , substituting the $L_1$ norm with an $L_2$ norm results in a degradation of task performance (Win Rate) across 9 out of 10 environments. Mathematically, the $L_1$ norm is well-suited for this objective because it naturally induces exact sparsity, effectively driving the values of redundant messages strictly to zero. In contrast, the $L_2$ norm tends to distribute the penalty, uniformly shrinking message values without discretely eliminating them. Consequently, the $L_2$ variant struggles to effectively sparsify communication in several environments (e.g., Hallway, where the communication rate balloons to 0.31). While the $L_2$ norm does achieve slightly lower communication rates on a few specific maps (such as \texttt{1o\_10b\_vs\_1r} and \texttt{Zerg 10\_vs\_10}), it does so indiscriminately, sacrificing massive amounts of task performance. This confirms that the $L_1$ norm is essential for safely pruning the communication channel without destroying cooperative coordination.

\section{Empirical Validation of Bounded Estimation Errors}
\label{app:estimation_errors}

In~\cref{thm:return_gap}, the performance bound of the unlearned policy relies on two key assumptions: the action-value approximation error is bounded by $\epsilon$, and the Message Value Estimator (MVE) error is bounded by $\delta$. To demonstrate that these theoretical assumptions hold in practice, we tracked the empirical max-norm deviation ($\|Q^*_{\text{tot}} - Q_{\text{tot}}\|_\infty$) and the MVE loss throughout the unlearning phase.

As summarized in~\cref{tab:max_norm_q_tot}, the empirical errors remain exceptionally tight across all evaluated environments. The max-norm deviation, which corresponds to $\epsilon$, stays consistently below $1.0$, ranging from $0.15$ on \texttt{MMM2} to $0.48$ on \texttt{1o\_2r\_vs\_4r}. Similarly, the MVE loss, which relates to the estimation error $\delta$, remains low and stable across the benchmarks. Furthermore, the sparsity loss is consistently near zero, indicating stable convergence during the unlearning process. 

These empirical results strongly validate our theoretical assumptions. They confirm that the value estimations do not diverge wildly when messages are unlearned, and that the anchoring objective ($\mathcal{L}_{\text{anchor}}$) successfully enforces the bounded errors required to guarantee safe unlearning.

\begin{table}[H]
\centering
\caption{Summary of $\|Q_{\text{tot}}^* - Q_{\text{tot}}\|_\infty$, MVE loss, and Sparsity loss per episode across 10 environments during the unlearning phase.}
\label{tab:max_norm_q_tot}
    \resizebox{0.5\linewidth}{!}{%
    \tablestyle{1.5pt}{1.2}
\begin{tabular}{lrrr}
\toprule
Env & $\|Q_{\text{tot}}^* - Q_{\text{tot}}\|_\infty$ & MVE Loss & Sparsity Loss \\
\midrule
1o\_10b\_vs\_1r & $0.46 \pm 0.18$ & $0.52 \pm 0.06$ & $0.04 \pm 0.04$ \\
1o\_2r\_vs\_4r & $0.48 \pm 0.06$ & $0.38 \pm 0.03$ & $0.01 \pm 0.01$ \\
MMM2 & $0.15 \pm 0.03$ & $0.18 \pm 0.09$ & $0.02 \pm 0.01$ \\
corridor & $0.39 \pm 0.11$ & $0.39 \pm 0.51$ & $0.06 \pm 0.04$ \\
hallway & $0.15 \pm 0.05$ & $0.41 \pm 0.05$ & $0.01 \pm 0.01$ \\
hallway\_group & $0.44 \pm 0.05$ & $0.41 \pm 0.15$ & $0.01 \pm 0.00$ \\
terran\_10v10 & $0.43 \pm 0.17$ & $0.37 \pm 0.13$ & $0.02 \pm 0.01$ \\
tj\_hard & $0.40 \pm 0.24$ & $0.17 \pm 0.09$ & $0.00 \pm 0.00$ \\
tj\_medium & $0.46 \pm 0.25$ & $0.06 \pm 0.06$ & $0.00 \pm 0.00$ \\
zerg\_10v10 & $0.43 \pm 0.15$ & $0.26 \pm 0.11$ & $0.03 \pm 0.01$ \\
\bottomrule
\end{tabular}
}
\end{table}

\section{Relation to Learned Empty Messages}
\label{app:nullmsg}

\begin{figure}[H]
\centering
\includegraphics[width=0.9\linewidth]{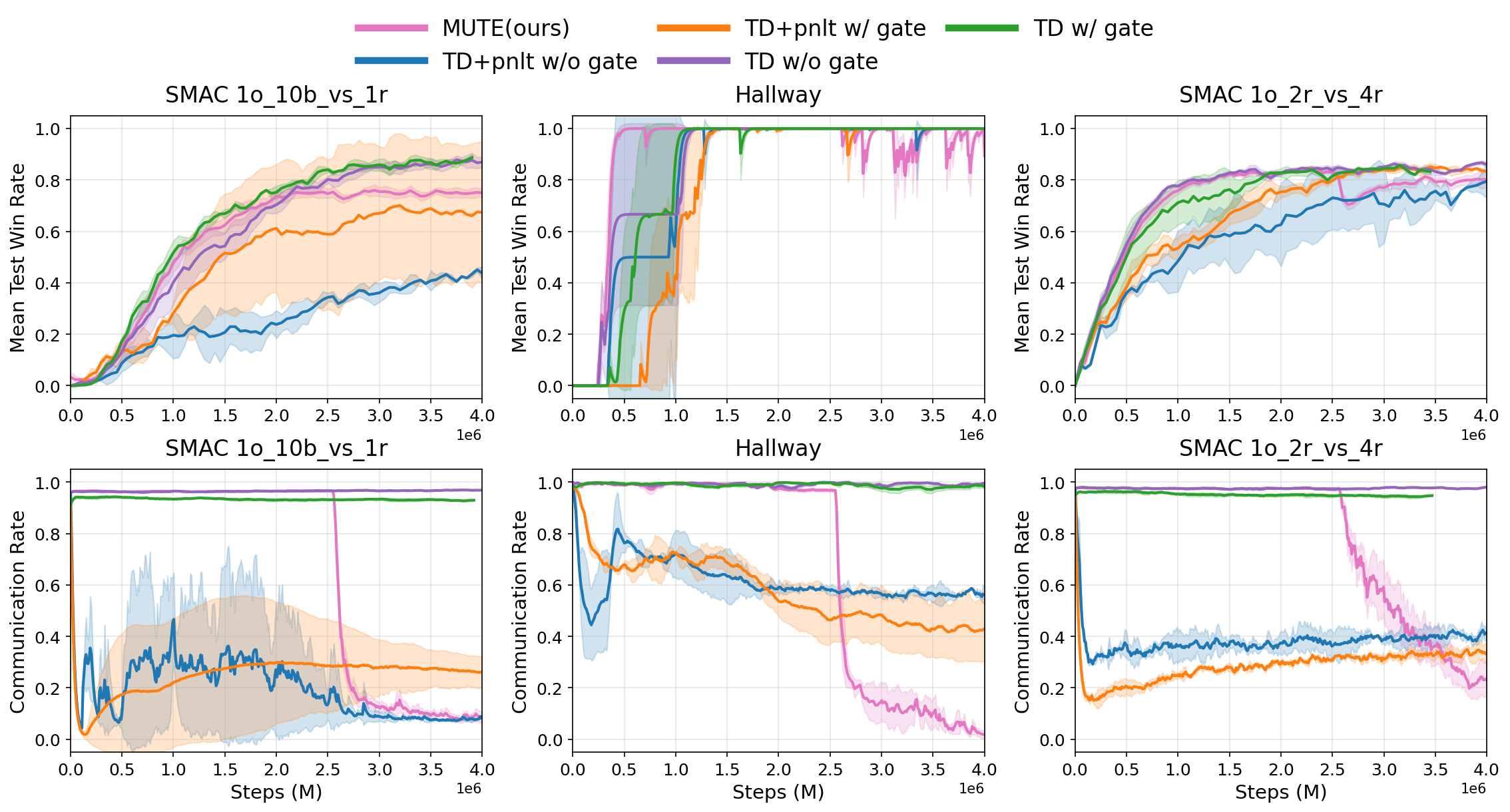}
\caption{Training curves of Win Rate and Communication Rate between the MUTE and }
\label{fig:nullms}
\vspace{-1em}
\end{figure}

\begin{table}[H]
    \small
    \centering
    \caption{Combined Win Rate and Communication Rate (mean $\pm$ CI) for MUTE vs. Baseline Methods across Different Environments.}
    \label{tab:mute_full_combined_results}
    \resizebox{\linewidth}{!}{%
    \tablestyle{1.5pt}{1.2}
    \begin{tabular}{lcccccccccc}
        \toprule
        & \multicolumn{5}{c}{\bfseries Test Win Rate (\%)} & \multicolumn{5}{c}{\bfseries Communication Rate (\%)} \\
        \cmidrule(lr){2-6} \cmidrule(lr){7-11}
        {\bfseries Environment} & {\bfseries MUTE (ours)} & {\bfseries TD+pnlt w/o gate} & {\bfseries TD+pnlt w/ gate} & {\bfseries TD w/o gate} & {\bfseries TD w/ gate} & {\bfseries MUTE (ours)} & {\bfseries TD+pnlt w/o gate} & {\bfseries TD+pnlt w/ gate} & {\bfseries TD w/o gate} & {\bfseries TD w/ gate} \\
        \midrule
        smac\_1o\_10b\_vs\_1r & $0.77 \pm 0.02$ & $0.44 \pm 0.02$ & $0.67 \pm 0.28$ & $0.87 \pm 0.02$ & $0.89 \pm 0.01$ & $0.06 \pm 0.01$ & $0.09 \pm 0.00$ & $0.26 \pm 0.06$ & $0.97 \pm 0.00$ & $0.93 \pm 0.00$ \\
        hallway & $0.97 \pm 0.09$ & $1.00 \pm 0.00$ & $1.00 \pm 0.00$ & $1.00 \pm 0.00$ & $1.00 \pm 0.00$ & $0.02 \pm 0.02$ & $0.56 \pm 0.04$ & $0.43 \pm 0.13$ & $1.00 \pm 0.00$ & $0.98 \pm 0.01$ \\
        smac\_1o\_2r\_vs\_4r & $0.81 \pm 0.03$ & $0.79 \pm 0.07$ & $0.83 \pm 0.00$ & $0.86 \pm 0.01$ & $0.83 \pm 0.01$ & $0.23 \pm 0.07$ & $0.41 \pm 0.03$ & $0.34 \pm 0.03$ & $0.98 \pm 0.00$ & $0.94 \pm 0.01$ \\
        \bottomrule
    \end{tabular}%
    }
    \vspace{-1em}
\end{table}

A natural question arises: instead of our proposed unlearning framework, could standard RL simply learn to output an empty message (e.g., an all-zero token or via a learned gating mechanism) to reduce communication overhead? In multi-agent settings heavily reliant on shared, delayed team rewards, the benefit of silencing one specific message is weakly attributable. Consequently, standard RL struggles to reliably discover \textit{selective silence} purely through environmental reward feedback, even when silencing is explicitly provided as an available action. 

To verify this, \cref{tab:mute_full_combined_results} compares MUTE against standard RL baselines augmented with communication gating (\texttt{w/ gate}) and direct communication penalties in the loss function (\texttt{+pnlt}).
Notably, while the baseline variants allocate their entire 4M-step budget to standard RL training, MUTE achieves its results under the exact same budget by partitioning it into 2M steps for RL pre-training, 0.5M steps for MVE training, and 1.5M steps for value-guided unlearning.

As the empirical results illustrate, simply providing an empty-message option (\texttt{TD w/ gate}) does not reduce the communication bandwidth; both \texttt{TD w/ gate} and \texttt{TD w/o gate} keep the communication rate near 100\%.
Conversely, while adding a direct penalty (\texttt{TD+pnlt w/o gate} and \texttt{TD+pnlt w/ gate}) forces a reduction in communication, the suppression is less targeted. This untargeted penalization degrades coordination and results in slower convergence. 
In contrast, MUTE's counterfactual unlearning allows it to maintain highly competitive win rates while achieving the lowest overall communication rate.

\section{Compuatational Cost and Scalability}
\label{app:comp_profiling}

\begin{table}[H]
    \centering
    \caption{Comprehensive computational profiling. (a) compares baseline training and inference costs. (b), (c), and (d) ablate the computational cost of evaluating Eq. (3) in \texttt{traffic\_junction} across different hyperparameters.}
    \label{tab:combined_profiling}

    \begin{subtable}{0.48\linewidth}
        \centering
        \caption{Computational profiling in SMAC 1o\_10b\_vs\_1r.}
        \label{tab:comp_profiling}
        \resizebox{\linewidth}{!}{%
        \tablestyle{1.5pt}{1.2}
        \begin{tabular}{lccc}
            \toprule
            \textbf{Method} & \textbf{Parameters} & \textbf{Training Time (hrs)} & \textbf{Inference FLOPs (M)} \\
            \midrule
            QMIX   & 85,160  & $6.07 \pm 0.91$ & 0.79 \\
            IC3Net & 47,690  & $5.23 \pm 1.23$ & 0.35 \\
            TarMAC & 53,930  & $5.42 \pm 0.75$ & 1.11 \\
            MAIC   & 114,607 & $15.6 \pm 3.56$ & 2.48 \\
            MASIA  & 176,029 & $7.82 \pm 2.33$ & 1.36 \\
            MAPPO  & 85,200  & $2.03 \pm 0.15$ & 0.26 \\
            SMS    & 42,524  & $8.45 \pm 1.27$ & 0.84 \\
            MUTE   & 185,149 & $8.09 \pm 0.88$ & 1.55 \\
            \bottomrule
        \end{tabular}%
        }
    \end{subtable}\hfill
    \begin{subtable}{0.48\linewidth}
        \centering
        \caption{Cost per timestep across agent numbers ($N$).}
        \label{tab:agt_num_sweep}
        \resizebox{\linewidth}{!}{%
        \begin{tabular}{lrrr}
            \toprule
            \textbf{$N$} & \textbf{Time (ms)} & \textbf{MFLOPs} & \textbf{Mean Test Win Rate} \\
            \midrule
            2  & $3.077 \pm 0.03$  & 23.6   & $1.000 \pm 0.05$ \\
            3  & $5.405 \pm 0.05$  & 47.7   & $0.975 \pm 0.04$ \\
            5  & $8.053 \pm 0.05$  & 123.2  & $0.975 \pm 0.04$ \\
            7  & $8.802 \pm 0.01$  & 238.7  & $0.875 \pm 0.10$ \\
            10 & $14.589 \pm 0.09$ & 497.2  & $0.725 \pm 0.07$ \\
            15 & $27.791 \pm 0.04$ & 1195.9 & $0.850 \pm 0.05$ \\
            20 & $30.471 \pm 0.11$ & 2301.6 & $0.625 \pm 0.08$ \\
            \bottomrule
        \end{tabular}%
        }
    \end{subtable}

    \vspace{1.5em} %

    \begin{subtable}{0.48\linewidth}
        \centering
        \caption{Cost per timestep across message dimensions ($d_{\text{msg}}$).}
        \label{tab:msg_dim_sweep}
        \resizebox{\linewidth}{!}{%
        \begin{tabular}{lrrr}
            \toprule
            \textbf{$d_{\text{msg}}$} & \textbf{Time (ms)} & \textbf{MFLOPs} & \textbf{Mean Test Win Rate} \\
            \midrule
            4  & $10.90 \pm 0.25$ & 458.17  & $0.85 \pm 0.02$ \\
            8  & $14.80 \pm 0.91$ & 497.23  & $0.95 \pm 0.03$ \\
            16 & $15.17 \pm 3.14$ & 575.35  & $0.96 \pm 0.01$ \\
            32 & $16.73 \pm 1.14$ & 731.58  & $0.90 \pm 0.05$ \\
            64 & $20.00 \pm 2.56$ & 1044.04 & $0.68 \pm 0.03$ \\
            \bottomrule
        \end{tabular}%
        }
    \end{subtable}\hfill
    \begin{subtable}{0.48\linewidth}
        \centering
        \caption{Cost across number of sampled experiences ($N_{\text{exp}}$).}
        \label{tab:seq_len_sweep}
        \resizebox{\linewidth}{!}{%
        \begin{tabular}{lrrr}
            \toprule
            \textbf{$N_{\text{exp}}$} & \textbf{Time (s)} & \textbf{GFLOPs} & \textbf{Mean Test Win Rate} \\
            \midrule
            4  & 0.4346 & 2.49  & $0.78 \pm 0.01$ \\
            8  & 0.4667 & 4.97  & $0.88 \pm 0.01$ \\
            16 & 0.5772 & 9.94  & $0.90 \pm 0.01$ \\
            32 & 0.5632 & 19.89 & $0.73 \pm 0.01$ \\
            64 & 0.6156 & 39.78 & $0.88 \pm 0.02$ \\
            \bottomrule
        \end{tabular}%
        }
    \end{subtable}

\end{table}

Because MUTE introduces an explicit Message Value Estimator (MVE) and a dedicated unlearning phase, it is critical to evaluate its computational overhead.
All computational profiling and timing evaluations were conducted on an L40S GPU.
\Cref{tab:combined_profiling}(a) compares the training wall-clock time and inference FLOPs of MUTE against established baselines in the SMAC benchmark.
While MUTE introduces additional parameters (185K) to support the MVE, its total training time (8.09 hours) remains highly competitive.
It trains significantly faster than heavy communication frameworks like MAIC (15.6 hours) and is on par with MASIA and SMS. 
During inference, MUTE operates with 1.55M FLOPs per step, confirming that the value-guided unlearning process does not impose prohibitive forward-pass bottlenecks compared to other state-of-the-art multi-agent methods.

To understand how the counterfactual value estimation scales in denser environments, we ablate the cost of evaluating~\cref{eq:cmv} across a varying number of agents ($N$) in \texttt{traffic\_junction}. As shown in \cref{tab:combined_profiling}(b), the computational time and MFLOPs scale as the environment grows from 2 to 20 agents. However, even at $N=20$, the evaluation time remains practically efficient at approximately 30 ms per timestep. The mean test win rate slightly decays at extreme agent counts, which reflects the inherent environmental difficulty of large-scale traffic coordination rather than a computational bottleneck.
Finally, \cref{tab:combined_profiling}(c) and (d) analyze the scalability of MUTE with respect to its internal hyperparameters: message dimension ($d_{\text{msg}}$) and the number of sampled experiences ($N_{\text{exp}}$). Increasing $d_{\text{msg}}$ yields a sub-linear growth in evaluation time, but we observe that excessively large message spaces (e.g., $d_{\text{msg}} = 64$) cause a drop in the win rate, likely due to the increased difficulty of accurately estimating counterfactual values in high-dimensional continuous spaces. Similarly, sweeping $N_{\text{exp}}$ demonstrates that GFLOPs scale strictly linearly with the sample size. A moderate configuration (e.g., $N_{\text{exp}} = 16$) provides an optimal balance, yielding peak coordination performance without incurring unnecessary computational drag.

\section{Normalized Performance--Communication Tradeoff}
\label{app:normalized_tradeoff}

\begin{figure}[H]
\centering
\includegraphics[width=\linewidth]{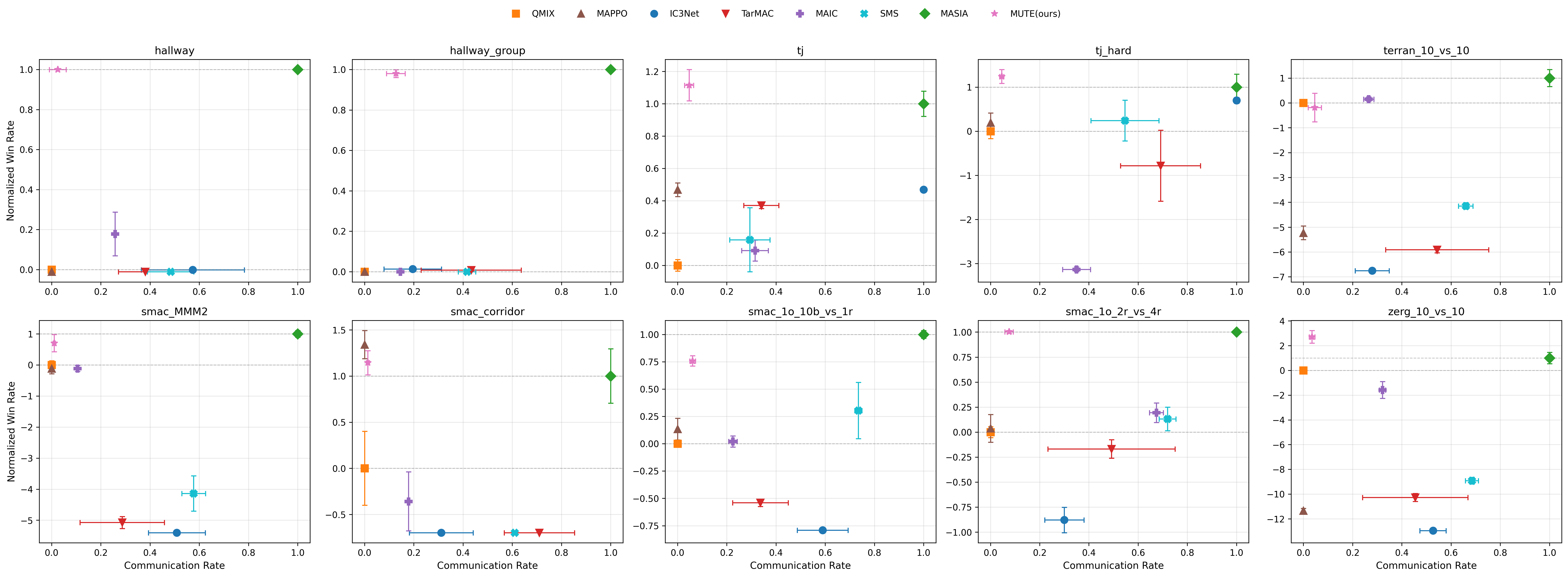}
\caption{Comparison of performance and communication }
\label{fig:scatter_comm_vs_norm_wi}
\vspace{-1em}
\end{figure}

To explicitly evaluate the tradeoff between coordination performance and communication bandwidth, we project the results of all evaluated methods onto a normalized performance--communication plane.

Specifically, for each environment, we calculate the normalized win rate ($W_{\text{norm}}$) as follows:
\begin{equation}
W_{\text{norm}} = \frac{W - W_{\text{base}}}{W_{\text{full}} - W_{\text{base}} + \epsilon}
\end{equation}
where $W$ is the raw win rate of the evaluated method, $W_{\text{base}}$ is the win rate of the no-communication baseline (QMIX), $W_{\text{full}}$ is the win rate of the full-communication baseline (MASIA), and $\epsilon$ is a small constant added for numerical stability to prevent division by zero. 

\Cref{fig:scatter_comm_vs_norm_wi} visualizes this normalized win rate against the final communication rate for 10 diverse environments. In this 2D plane, the ideal operating region is the top-left corner, which corresponds to matching or exceeding full-communication performance ($W_{\text{norm}} \ge 1.0$) while incurring minimal communication overhead (Communication Rate $\approx 0.0$). 

As the scatter plots clearly demonstrate, MUTE (represented by the pink star) consistently dominates this ideal top-left region across the benchmarks. While the full-communication baseline MASIA anchors the top-right ($W_{\text{norm}} = 1.0$, Communication Rate $= 1.0$), MUTE matches or occasionally exceeds this performance (e.g., in \textit{smac\_corridor} and \textit{tj\_hard}) while reducing communication by over 90\%. 

Furthermore, the plots highlight the failure modes of existing communication-reduction baselines. Methods such as IC3Net, TarMAC, and SMS frequently suffer from severe drops in normalized win rates (dropping well below $0.0$, indicating worse performance than the no-communication QMIX baseline) or fail to meaningfully reduce the communication rate. MUTE's value-guided counterfactual unlearning is the only approach that robustly preserves absolute coordination capabilities while successfully isolating and eliminating redundant messages.

\section{Ablation on Training Budget Allocation}
\label{app:budget_ablation}

To ensure that our reported three-stage training schedule was not chosen arbitrarily, we performed a budget-sensitivity ablation on the \textit{smac\_1o\_10b\_vs\_1r} environment. MUTE operates on a fixed total training budget of 4M environment steps, which must be partitioned across three distinct phases: (1) RL pre-training of the dense expert policy, (2) training the Message Value Estimator (MVE), and (3) value-guided unlearning. 

In this ablation, we compare our default schedule (2M / 0.5M / 1.5M) against several alternative allocations while keeping the 4M total step budget strictly constant.

\begin{figure}[H]
    \centering
    \includegraphics[width=0.75\linewidth]{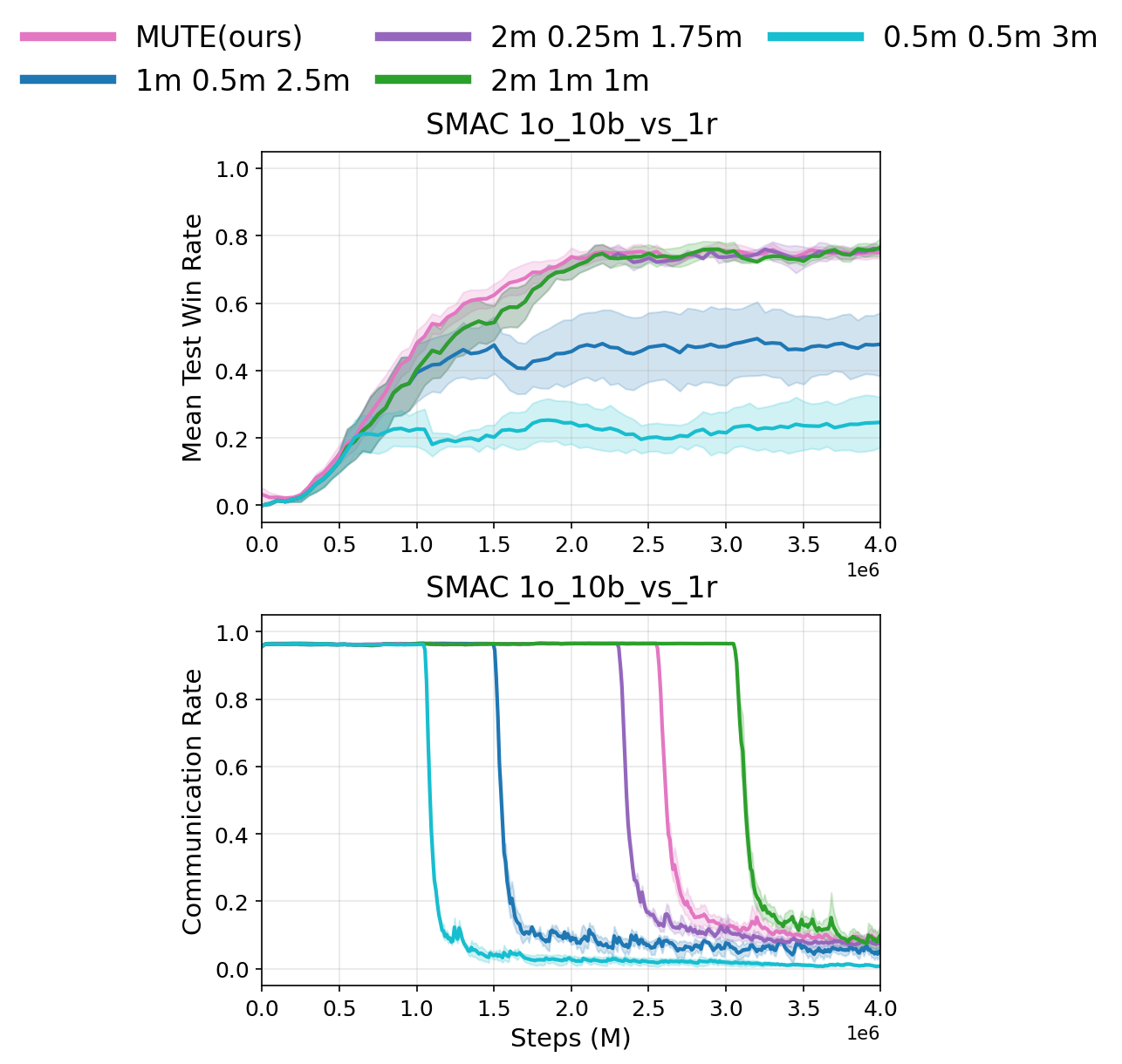} 
    \caption{Training curves of Mean Test Win Rate and Communication Rate for different budget allocations in \textit{smac\_1o\_10b\_vs\_1r}. The sharp drops in communication rate correspond to the initiation of the unlearning phase.}
    \label{fig:budget_ablation}
    \vspace{-1em}
\end{figure}

\begin{table}[H]
    \small
    \centering
    \caption{Final Win Rate and Communication Rate (mean $\pm$ CI) under different training budget allocations in \textit{smac\_1o\_10b\_vs\_1r}. Allocations are denoted in millions of steps as (Pre-training / MVE / Unlearning).}
    \label{tab:budget_ablation}
    \begin{tabular}{lcc}
        \toprule
        {\bfseries Allocation Schedule (M Steps)} & {\bfseries Mean Test Win Rate (\%)} & {\bfseries Communication Rate (\%)} \\
        \midrule
        2.0 / 0.50 / 1.50 (MUTE default) & $0.77 \pm 0.02$ & $0.06 \pm 0.01$ \\
        1.0 / 0.50 / 2.50                & $0.48 \pm 0.09$ & $0.05 \pm 0.02$ \\
        2.0 / 0.25 / 1.75                & $0.76 \pm 0.01$ & $0.09 \pm 0.01$ \\
        2.0 / 1.00 / 1.00                & $0.77 \pm 0.03$ & $0.10 \pm 0.03$ \\
        0.5 / 0.50 / 3.00                & $0.25 \pm 0.07$ & $0.01 \pm 0.00$ \\
        \bottomrule
    \end{tabular}
\end{table}

As shown in \cref{tab:budget_ablation} and \cref{fig:budget_ablation}, the results demonstrate a clear and consistent trend regarding the onset of the unlearning phase. Triggering unlearning too early (e.g., the 0.5M / 0.5M / 3M or 1M / 0.5M / 2.5M schedules) substantially degrades the final coordination performance. This indicates that the dense expert policy has not yet sufficiently converged; attempting to isolate and unlearn redundant messages before the foundational communication strategy is stable leads to catastrophic performance drops.

In contrast, schedules that allocate a full 2M steps to initial pre-training perform much better overall. Among these stable configurations, altering the ratio of MVE training to unlearning (e.g., 2M / 0.25M / 1.75M or 2M / 1M / 1M) yields similar win rates, but slightly inferior communication reduction. Our empirically chosen default allocation (2M / 0.5M / 1.5M) strikes the optimal balance: it achieves one of the strongest final win rates while simultaneously reaching a highly efficient, near-zero final communication rate.

\end{document}